%% file: paper.tex
\documentclass[letterpaper,twocolumn,10pt]{article}
\usepackage{usenix2019,epsfig,endnotes}

\usepackage{amsmath}
\usepackage{amssymb}
\usepackage{amsthm}
\usepackage{bm}
\usepackage{graphicx}
\PassOptionsToPackage{hyphens}{url}
\usepackage{url}
\usepackage{color}    
\usepackage{xcolor}
\usepackage{cite}
\usepackage[ruled,vlined]{algorithm2e}
\usepackage[bookmarks=false]{hyperref}
\usepackage{listings}
\usepackage{paralist}
\usepackage{multicol}
\usepackage{amsfonts}
\usepackage{commath}
\usepackage{xspace}
\usepackage{subfig}
\usepackage{flushend}

\usepackage[font={small},skip=3pt,belowskip=-7pt]{caption}
\usepackage{url}
\usepackage{caption}
\usepackage{float}
\usepackage{multirow}
\usepackage{array}
\usepackage{fmtcount}
\usepackage{footmisc}
\usepackage{arydshln}

\usepackage{booktabs} 
\usepackage{colortbl} 
\usepackage{xfrac}

\usepackage{tikz}
\usepackage{pgfplots, pgfplotstable}
\usepgfplotslibrary{fillbetween}
\usetikzlibrary{patterns}

\usepackage{pifont}

\newtheorem{lemma}{Lemma}

\usepackage{adjustbox}

\newcommand{\myparagraph}[1]{\noindent{\bf #1}}

\begin{document}

\date{}

\newcommand{\name}{StrongChain\xspace}%
\title{\Large \bf \name: Transparent and Collaborative Proof-of-Work Consensus}

\author{
    {\rm Pawel Szalachowski$^1$ \hspace{0.7cm} Dani\"el Reijsbergen$^1$
    \hspace{0.7cm} Ivan Homoliak$^1$ \hspace{0.7cm} Siwei
    Sun$^{2,}$\thanks{This work was done while the author was at SUTD.}}\\ 
    $^1$Singapore University of Technology and Design (SUTD)\\
    $^2$Institute of Information Engineering and DCS Center, Chinese Academy of Sciences
}

\maketitle

\thispagestyle{empty}

\subsection*{Abstract}
Bitcoin is the most successful cryptocurrency so far.  This is mainly due to its
novel consensus algorithm, which is based on proof-of-work combined with a
cryptographically-protected data structure and a rewarding scheme that
incentivizes nodes to participate.  However, despite its unprecedented success
Bitcoin suffers from many inefficiencies. For instance, Bitcoin's consensus
mechanism has been proved to be incentive-incompatible, its high reward variance
causes centralization, and its hardcoded deflation raises questions about its
long-term sustainability.  

In this work, we revise the Bitcoin consensus mechanism by proposing \name, a
scheme that introduces transparency and incentivizes participants to collaborate
rather than to compete.  The core design of our protocol is to reflect and
utilize the computing power aggregated on the blockchain which is invisible and
``wasted'' in Bitcoin today.  Introducing relatively easy, although important
changes to Bitcoin's design enables us to improve many crucial aspects of
Bitcoin-like cryptocurrencies making it more secure, efficient, and profitable
for participants.  We thoroughly analyze our approach and we present an
implementation of \name.  The obtained results confirm its efficiency, security,
and deployability.

\section{Introduction}
\label{sec:intro}
\input{sec/intro}

\section{Background and Problem Definition}
\subsection{Nakamoto Consensus and Bitcoin}
\label{sec:pre}
\input{sec/pre}
\label{sec:problem}
\input{sec/problem}

\label{sec:details}
\input{sec/details}

\section{Analysis}
\label{sec:analysis}
\input{sec/analysis}

\section{Discussion}
\label{sec:appendix}
\input{sec/appendix}

\section{Realization in Practice}
\label{sec:implementation}
\input{sec/implementation}


\section{Related work}
\label{sec:related}
\input{sec/related}

\section{Conclusions}
\label{sec:conclusions}
\input{sec/conclusions}

\section*{Acknowledgment}
We thank the anonymous reviewers and our shepherd Joseph Bonneau for their
valuable comments and suggestions.

This work was supported in part by the National Research Foundation (NRF), Prime
Minister's Office, Singapore, under its National Cybersecurity R\&D Programme
(Award No. NRF2016NCR-NCR002-028) and administered by the National Cybersecurity
R\&D Directorate, and by ST Electronics and NRF under Corporate Laboratory @
University Scheme (Programme Title: STEE Infosec - SUTD Corporate Laboratory).

{\normalsize \bibliographystyle{abbrv}
\bibliography{ref}}


\end{document}

%% file: sec/intro.tex
One of the main novelties of Bitcoin~\cite{nakamoto2008bitcoin} is Nakamoto
consensus. This mechanism enabled the development of a permissionless, anonymous, and
Internet-scale consensus protocol,  and combined with incentive mechanisms
allowed Bitcoin to emerge as the first decentralized cryptocurrency.  Bitcoin is
successful beyond all expectations, has inspired many other projects, and has
started new research directions.  Nakamoto consensus is based on
proof-of-work (PoW)~\cite{dwork1992pricing} in order to mitigate Sybil
attacks~\cite{douceur2002sybil}. To prevent modifications, a
cryptographically-protected append-only list~\cite{bayer1993improving} is
introduced. This list consists of  transactions grouped into blocks and is
usually referred to as a \textit{blockchain}.  Every active protocol
participant (called a \textit{miner}) collects transactions sent by users and
tries to solve a computationally-hard puzzle in order to be able to write to the
blockchain (the process of solving the puzzle is called \textit{mining}).  When
a valid solution is found, it is disseminated along with the transactions that
the miner wishes to append.  Other miners verify this data and, if valid, append
it to their replicated blockchains.  The miner that has found a solution is
awarded by a) the system, via a rewarding scheme programmed into the protocol,
and b) fees paid by transaction senders. All monetary transfers
in Bitcoin are expressed in its native currency (called bitcoin, abbreviated as
BTC) whose supply is limited by the protocol.

Bitcoin has started an advent of decentralized cryptocurrency systems and as the
first proposed and deployed system in this class is surprisingly robust.
However, there are multiple drawbacks of Bitcoin that undermine its security
promises and raise questions about its future.  Bitcoin has been proved to be
incentive-incompatible~\cite{eyal2015miner,sapirshtein2016optimal,eyal2014majority,gapgame}.
Namely, in some circumstances, the miners' best strategy is to not announce their
found solutions immediately, but instead withhold them for some time period.
Another issue is that the increasing popularity of the system tends towards its
centralization.  Strong competition between miners resulted in a high reward variance,
thus to stabilize their revenue miners started grouping their computing power
by forming \textit{mining pools}.  Over time, mining pools have come to dominate the
computing power of the system, and although they are beneficial for miners,
large mining pools are risky for the system as they have multiple ways of
abusing the
protocol~\cite{karame2012double,eyal2015miner,eyal2014majority,sapirshtein2016optimal}.
Recently, researchers rigorously analyzed one of the impacts of Bitcoin's
deflation~\cite{longrunBitcoin,instability_noReward,gapgame}. Their results
indicate that Bitcoin may be unsustainable in the long term, mainly due to
decreasing miners' rewards that will eventually stop completely. Besides that,
unusually for a transaction system, Bitcoin is designed to favor availability
over consistency. This choice was motivated by its open and permissionless
spirit, but in the case of inconsistencies (i.e., \textit{forks} in the
blockchain) the system can be slow to converge.

Motivated by these drawbacks, we propose \name, a simple yet powerful revision
of the Bitcoin consensus mechanism.  Our main intuition is to design
a system such that the mining process is more transparent and collaborative,
i.e., miners get better knowledge about the mining power of the system and they are incentivized to solve puzzles together rather than compete.  In order to achieve
it, in the heart of the \name's design we employ \textit{weak solutions}, i.e.,
puzzle solutions with a PoW that is significant yet insufficient for a standard
solution.  We design our system, such that a) weak solutions are part of the
consensus protocol, b) their finders are rewarded independently, and c) miners
have incentives to announce own solutions and append solutions of others
immediately.  We thoroughly analyze our approach and show that with these
changes, the mining process is becoming more transparent, collaborative, secure,
efficient, and decentralized.  Surprisingly, we also show how our approach can
improve the freshness properties offered by Bitcoin.  We present an
implementation and evaluation of our scheme.

%% file: sec/pre.tex
\label{sec:pre:bitcoin}
The Nakamoto consensus protocol allows decentralized and distributed network
comprised of mutually distrusting participants to reach an agreement on the
state of the global distributed ledger~\cite{nakamoto2008bitcoin}.  The
distributed ledger can be regarded as a linked list of blocks, referred to as
the {\it blockchain}, which serializes and confirms ``transactions''.  To
resolve any \textit{forks} of the blockchain the protocol specifies to always
accept the longest chain as the current one.  
Bitcoin is a peer-to-peer cryptocurrency that deploys Nakamoto consensus as
its core mechanism to avoid double-spending. 
Transactions spending bitcoins are announced to
the Bitcoin network, where miners validate, serialize all non-included
transactions, and try to create (mine) a block of transactions with a PoW
embedded into the block header.  
A valid block must fulfill the condition that for a cryptographic hash function
$H$, the hash value of the block header is less than the target $T$.

Brute-forcing the nonce (together with some other changeable data fields) is
virtually the only way to produce the PoW, which costs computational resources
of the miners.
To incentivize miners, the Bitcoin protocol allows the miner who finds a block
to insert a special transaction (see below) minting a specified amount of new
bitcoins and collecting transaction fees offered by the included transactions,
which are transferred to an account chosen by the miner.
Currently, every block mints 12.5 new bitcoins. This amount is halved 
every four years, upper-bounding the number of bitcoins that will be created to
a fixed total of 21 million coins. It implies that after around the year 2140, 
no new coins will be created, and the transaction fees will be the only source of
reward for miners. Because of its design, Bitcoin is a deflationary currency. 

The overall hash rate of the Bitcoin network  and the difficulty of the PoW
determine how long it takes to generate a new block for the whole network (the
block interval).  To stabilize the block interval at about 10 minutes for
the constantly changing total mining power, the Bitcoin network adjusts the target $T$ every
$2016$ blocks (about two weeks, i.e., a \textit{difficulty window}) according to
the following formula
\begin{equation}\label{eqn:adjust_Ts}
T_{new} = T_{old} 
\cdot \frac{\textit{Time~of~the~last~2016~blocks}}{\textit{2016}\cdot\textit{10~minutes}}.
\end{equation}
In simple terms, the difficulty increases if the network is finding blocks
faster than every 10 minutes, and decrease otherwise.  With dynamic difficulty, Nakamoto's longest chain rule was considered as a
bug,\footnote{\url{https://goo.gl/thhusi}} as it is trivial to produce long chains
that have low difficulty. The rule was replaced by the strongest-PoW chain rule
where competing chains are measured in terms of PoW they aggregated.  As long as
there is one chain with the highest PoW, this chain is chosen as the current
one.

Bitcoin introduced and uses the \textit{unspent transaction output} model.  The
validity of a Bitcoin transaction is verified by executing a script proving that
the transaction sender is authorized to redeem unspent coins.
The only exception is the first transaction in the transaction list
of a block, which implements how the newly minted bitcoins and transaction fees
are distributed.  It is called a \textit{coinbase transaction} and it contains the amount
of bitcoins (the sum of newly minted coins and the fees derived from all the
transactions) and the beneficiary (typically the creator of the block).  Also,
the Bitcoin scripting language offers a mechanism (\texttt{OP\_RETURN}) for
recording data on the blockchain, which facilitates third-party applications
built-on Bitcoin. 

Bitcoin proposes the simplified payment verification (SPV) protocol, that allows
resource-limited clients to verify that a transaction is indeed included in a
block provided only with the block header and a short transaction's inclusion
proof. The key advantage of the protocol is that SPV clients can verify the
existence of a transaction without downloading or storing the whole block.  SPV
clients are provided only with block headers and on-demand request from the
network inclusion proofs of the transactions they are interested in.
%
%
%
%

In the original white paper, Nakamoto heuristically argues that the consensus
protocol remains secure as long as a majority ($>50\%$) of the participants'
computing power honestly follow the rule specified by the protocol, which is
compatible with their own economic incentives.

%% file: sec/problem.tex
\subsection{Bitcoin Mining Issues}
\label{sec:problem:issues}
Despite its popularity, Nakamoto consensus and Bitcoin suffer from multiple
issues.  Bitcoin mining is not always incentive-compatible. By deviating from the
protocol and strategically withholding found blocks, a miner in possession of a
proportion $\alpha$ of the total computational power may occupy more than
$\alpha$ portion of the blocks on the blockchain, and therefore gain
disproportionally higher payoffs with respect to her
share~\cite{bahack2013theoretical,eyal2014majority,sapirshtein2016optimal}.
More specifically, an attacker tries to create a private chain by keeping found
blocks secret as long as the chain is in an advantageous position with one or
more blocks more than the public branch.  She releases her private chain only when
the public chain has almost caught up, hence invalidating the public branch and
all the efforts made by the honest miners.   This kind of attack, called
\emph{selfish mining}, can be more efficient when a well-connected selfish miner's
computational power exceeds a certain threshold (around more than 30\%).
Thus, selfish mining does not pay off if the
mining power is sufficiently decentralized.

Unfortunately, the miners have an impulse to centralize their computing
resources due to Bitcoin's rewarding scheme.  In Bitcoin, rewarding is a
zero-sum game and only the lucky miner who manages to get her block accepted
receives the reward, while others who indeed contributed computational resources
to produce the PoW are completely invisible and ignored.  Increasing mining
competition leads to an extremely high variance of the payoffs of a miner with a
limited computational power. A solo miner may need to wait months or years to
receive any reward at all.  As a consequence, miners are motivated to group
their resources and form mining pools, that divide work among pool participants
and share the rewards according to their contributions.  As of November 2018,
only five largest pools account for more than 65\% of the mining power of the
whole Bitcoin
network.\footnote{\url{https://btc.com/stats/pool?pool_mode=month}} Such mining
pools not only undermine the decentralization property of the system but also
raise various in-pool or cross-pool security
issues~\cite{rosenfeld2011analysis,courtois2014subversive,eyal2015miner,luu2015power}.

Another seemingly harmless characteristic of Bitcoin is its finite monetary
supply.  However, researchers in their recent
work~\cite{longrunBitcoin,instability_noReward,gapgame} investigate the system
dynamics when incentives coming from transaction fees are non-negligible
compared with block rewards (in one extreme case the incentives come only from
fees). They provide analysis and evidence, indicating an undesired system
degradation due to the rational and self-interested participants.  Firstly, such
a system incentivizes large miner coalitions, increasing the system centralization
even more. Secondly, it leads to a mining gap where miners would avoid mining
when the available fees are insufficient. Even worse, rational miners tend to
mine on chains that do not include available transactions (and their fees),
rather than following the block selection rule specified by the protocol,
resulting in a backlog of transactions. Finally, in the sole transaction fee
regime, selfish mining attacks are efficient for miners with arbitrarily low
mining power, regardless of their network connection qualities.  These results
suggest that making the block reward permanent and accepting the monetary
inflation may be a wise design choice to ensure the stability of the
cryptocurrency in the long run.  

Moreover, the chain selection rule (i.e., the strongest chain is accepted),
together with the network delay, occasionally lead to forks, where two or more
blocks pointing to the same block are created around the same time, causing the
participants to have different views of the current system state. Such
conflicting views will eventually be resolved since with a high probability one
branch will finally beat the others (then the blocks from the ``losing'' chain
become {\it stale blocks}).  The process of fork resolution is quite slow, as
blocks have the same PoW weight and they arrive in 10-minutes intervals (on
average).

Finally, the freshness properties provided by Bitcoin are questionable.  By
design, the Bitcoin blockchain preserves the order of blocks and transactions,
however, the accurate estimation of time of these events is
challenging~\cite{SzCVCBT18}, despite the fact that each block has an associated timestamp.  A block's timestamp is accepted if a) it is greater than the median
timestamp of the previous eleven blocks, and b) it is less than the network
time plus two hours.\footnote{\url{https://en.bitcoin.it/wiki/Block_timestamp}}
This gives significant room for manipulation ---
in
theory, a timestamp can differ in hours from the actual time since it is largely determined
by a single block creator.  In fact, as time cannot be accurately determined
from the timestamps, the capabilities of the Bitcoin protocol as a
timestamping service are limited, which may lead to severe attacks
by itself~\cite{gervais2015tampering,culubas}.

\subsection{Requirements} \label{sec:requirements}
For the purpose of revising a consensus protocol of PoW blockchains in a secure,
well-incentivized, and seamless way, we define the following respective requirements:

\begin{compactitem}		
    \item \textbf{Security} -- the scheme should improve the security of
        Nakamoto consensus by mitigating known attack vectors and preventing new
        ones. In essence, the scheme should be incentive-compatible, such that
        miners benefit from following the consensus rules and have no gain from
        violating them.

    \item \textbf{Reward Variance} -- another objective is to minimize the
        variance in rewards. This requirement is crucial for decentralization since a high
        reward variance is the main motivation of individual miners to join centralized
        mining pools. Centralization is undesirable as large-enough
        mining pools can attack the Bitcoin protocol. 

    \item \textbf{Chain Quality} -- the scheme should provide a high chain
        quality, which usually is described using the two following properties.
        \begin{compactitem}
            \item \textbf{Mining Power Utilization} -- the ratio between the
                mining power on the main chain and the mining power of the
                entire blockchain network.  This property describes the performance of
                mining and its ideal value is 1, which denotes that all mining
                power of the system contributes to the ``official'' or ``canonical'' chain.  A
                high mining power utilization implies a low stale block rate.
            \item \textbf{Fairness} -- the protocol should be fair, i.e., a
                miner should earn rewards proportionally to the resources
                invested by her in mining. We denote a miner with $\alpha$ of
                the global mining power as an \textit{$\alpha$-strong miner}.
        \end{compactitem}

    \item \textbf{Efficiency and Practicality} -- the scheme should not
        introduce any significant computational, storage, or bandwidth
        overheads.  This is especially important since Bitcoin works as a
        replicated state machine, therefore all full nodes replicate data and
        the validation process.  In particular, the block validation time, its
        size, and overheads of SPV clients should be at least similar as today.
        Moreover, the protocol should not introduce any assumptions that would
        be misaligned with Bitcoin's spirit and perceived as unacceptable by
        the community. In particular, the scheme should not introduce any
        trusted parties and should not assume strong synchronization  of nodes
        (like global and reliable timestamps).
\end{compactitem}

%% file: sec/details.tex
\section{High-level Overview}
\subsection{Design Rationale}
Our first observation is that Bitcoin mining is not transparent. It is difficult
to quickly estimate the computing power of the different participants, because the only indicator
is the found blocks. After all, blocks arrive with a low frequency, and each block is
equal in terms of its implied computational power.  Consequently, the only
way of resolving forks is to wait for a stronger chain to emerge, which can be a time-consuming process.  A related issue is block-withholding-like
attacks (e.g., selfish mining) which are based on the observation that sometimes it is
profitable for an attacker to deviate from the protocol by postponing
the announcement of new solutions.  We see transparency as a helpful property
also in this context.  Ideally, non-visible (hidden) solutions should be
penalized, however, in practice it is challenging to detect and prove that a
solution was hidden.  We observe that an alternative way of mitigating these
attacks would be to promote visible solutions, such that with more computing
power aggregated around them they get stronger.  This would incentivize miners
to publish their solutions immediately, since keeping it secret may be too risky
as other miners could strengthen a competing potential (future) solution over
time.  Finally, supported by recent research
results~\cite{longrunBitcoin,instability_noReward,gapgame,eyal2014majority,sapirshtein2016optimal},
we envision that redesigning the Bitcoin reward scheme is unavoidable to keep the
system sustainable and more secure.  Beside the deflation issues (see
\autoref{sec:problem:issues}), the reward scheme in Bitcoin is a zero-sum game
rewarding only lucky miners and ignoring all effort of other participants.  That
causes fierce competition between miners and a high reward variance, which stimulates
miners to collaborate, but within mining pools, introducing more risk to the
system.  We aim to design a system where miners can benefit from
collaboration but without introducing centralization risks. 

\subsection{Overview}
Motivated by these observations, we see weak puzzle solutions, currently
invisible and ``wasted'' in Bitcoin,  as a promising direction.  Miners
exchanging them could make the protocol more transparent as announcing them
could reflect the current distribution of computational efforts on the network.
Furthermore, if included in consensus rules, they could give blocks a better
granularity in terms of PoW, and incentivize miners to collaborate.  In our
scheme, miners solve a puzzle as today but in addition to publishing solutions,
they exchange weak solutions too (i.e., almost-solved puzzles). The lucky miner
publishes her solution that embeds gathered weak solutions (pointing to the same
previous block) of other miners.  Such a published block better reflects the aggregated
PoW of a block, which in the case of a fork can indicate that more mining power
is focused on a given branch (i.e., actually it proves that more computing power
``believes'' that the given branch is correct).  Another crucial change is to
redesign the Bitcoin reward system, such that the finders of weak solutions are also
rewarded. Following lessons learned from mining pool attacks, instead of sharing
rewards among miners, our scheme rewards weak solutions proportionally to their
PoW contributed to a given block and all rewards are independent of other
solutions of the block.
(Note, that this change requires a Bitcoin \textit{hard fork}.)

There are a few intuitions behind these design choices.  First, a selfish miner
finding a new block takes a high risk by keeping this block secret.  This is 
because blocks 
have a better granularity due to honest miners exchanging partial solutions and strengthening their
prospective block, which in the case of a fork would be stronger than the older
block kept secret (i.e., the block of the selfish miner).  Secondly, miners are
actually incentivized to collaborate by a) exchanging their weak solutions, and
b) by appending weak solutions submitted by other miners.  For the former case,
miners are rewarded whenever their solutions are appended, hence keeping them
secret can be unprofitable for them.  For the latter case, a miner appending
weak solutions of others only increases the strength of her potential block, and
moreover, appending these solutions does not negatively influence  the miner's
potential reward.  Finally, our approach comes with another benefit.
Proportional rewarding of weak solutions decreases the reward variance, thus
miners do not have to join large mining pools in order to stabilize their
revenue. This could lead to a higher decentralization of mining power on the
network.

In the following sections, we describe details of our system, show its analysis,
and report on its implementation.

\section{\name Details}

\begin{algorithm}[t!]
    \caption{Pseudocode of \name functions.}
	\label{alg:all}

    \small

    \SetKwProg{func}{function}{}{}

    \func{mineBlock()}{
        $\mathit{weakHdrsTmp\leftarrow\emptyset}$\;
        \For{$nonce \in \{0,1,2,...\}$}{
            $hdr\leftarrow\mathit{createHeader(nonce)}$\;
            /* check if the header meets the strong target */\\
            $h_{tmp}\leftarrow H(hdr)$\;
            \If{$h_{tmp} < T_s$}{
                $\mathit{B\leftarrow createBlock(hdr,weakHdrsTmp,Txs)}$\;
                $\mathit{broadcast(B)}$\;
                \Return; /* signal to mine with the new block */\\
    		}
            /* check if the header meets the weak target */\\
            \If{$h_{tmp} < T_w$}{
                $\mathit{weakHdrsTmp.add(hdr)}$\;
                $\mathit{broadcast(hdr)}$\;
    		}
    	}
    }

    \func{onRecvWeakHdr($\mathit{hdr})$}{
        $h_w\leftarrow H(hdr)$\;
        \textbf{assert}($T_s \leq h_w<T_w$ \textbf{and} \textit{validHeader(hdr)})\;
        \textbf{assert}($\mathit{hdr.PrevHash ==  H(lastBlock.hdr)}$) \;
        $\mathit{weakHdrsTmp.add(hdr)}$\;
    }

    \func{rewardBlock($B$)}{
        /* reward block finder with $R$ */\\
        reward($B.hdr.Coinbase, R + B.TxFees$)\; 
        $w\leftarrow \gamma*T_s/T_w$; /* reward weak headers proportionally */\\
        \For{$\mathit{hdr \in B.weakHdrSet}$}{
            reward($hdr.Coinbase, w * c * R$)\;
        }
    }

   \func{validateBlock($B$)}{
       \textbf{assert}($H(B.hdr)<T_s$ \textbf{and} \textit{validHeader(B.hdr)})\;
       \textbf{assert}($\mathit{B.hdr.PrevHash ==  H(lastBlock.hdr)}$) \;
       \textbf{assert}($\mathit{validTransactions(B)}$)\;
       \For{$\mathit{hdr \in B.weakHdrSet}$}{
           \textbf{assert}($T_s \leq H(hdr) <T_w$ \textbf{and} \textit{validHeader(hdr)})\;
            \textbf{assert}($\mathit{hdr.PrevHash ==  H(lastBlock.hdr)}$)\;
       }
   }

    \func{chainPoW($chain$)}{
        $sum\leftarrow 0$\;
        \For{$B \in chain$}{
            /* for each block compute its aggregated PoW */\\
            $\mathit{T_s \leftarrow B.hdr.Target}$\;
            $sum\leftarrow sum + T_{max}/T_{s}$\;
            \For{$\mathit{hdr \in B.weakHdrSet}$}{
                $sum\leftarrow sum + T_{max}/T_w$\;
            }
        }
        \Return $sum$\;
    }

    \func{getTimestamp($B$)}{
        $\mathit{sumT\leftarrow B.hdr.Timestamp}$\;
        $\mathit{sumW\leftarrow 1.0}$\;
        /* average timestamp by the aggregated PoW */\\
        $w\leftarrow T_{s}/T_w$\;
        \For{$\mathit{hdr \in B.weakHdrSet}$}{
            $sumT\leftarrow sumT + w*hdr.Timestamp$\;
            $sumW\leftarrow sumW + w$\;
        }
        \Return $sumT/sumW$\;
    }
\end{algorithm}

\subsection{Mining}
As in Bitcoin, in \name miners authenticate transactions by collecting them into
blocks whose headers are protected by a certain amount of PoW.  A simplified
description of a block mining procedure in \name is presented as the
\textit{mineBlock()} function in \autoref{alg:all}.  Namely, every miner tries
to solve a PoW puzzle by computing the hash function over a newly created
header.  The header is constantly being changed by modifying its nonce
field,\footnote{In fact, other fields can be modified too if needed.} until a
valid hash value is found.  Whenever a miner finds a header $hdr$ whose hash
value $h=H(hdr)$ is smaller than the \textit{strong target} $T_s$, i.e., a $h$
that satisfies the following:
\begin{equation*}
    h < T_s,
\end{equation*}
then the corresponding block is announced to the network and becomes, with all
its transactions and metadata, part of the blockchain.  We refer to headers of
included blocks as \textit{strong headers}.

One of the main differences with Bitcoin is that our mining protocol handles also
headers whose hash values do not meet the strong target $T_s$, but still are low
enough to prove a significant PoW. We call such a header a \textit{weak header}
and its hash value $h$ has to satisfy the following:
\begin{equation}
    \label{eq:weak}
    T_s \leq h < T_w,
\end{equation}
where $T_w > T_s$ and $T_w$ is called the \textit{weak target}.

Whenever a miner finds such a block header, she adds it to her local list of
weak headers (i.e., \textit{weakHdrsTmp}) and she propagates the header among
all miners.  Then every miner that receives this information first validates it
(see \textit{onRecvWeakHdr()}) by checking whether
\begin{compactitem}
    \item the header points to the last strong header,
    \item its other fields are correct (see \autoref{sec:details:block:layout_valid}),
    \item and \autoref{eq:weak} is satisfied.
\end{compactitem}
Afterward, miners append the header to their lists of weak headers.
We do not limit the number of weak headers appended, although this number is
correlated with the $T_w/T_s$ ratio (see \autoref{sec:analysis}).

Finally, miners continue the mining process in order to find a strong header.
In this process, a miner keeps creating candidate headers by computing hash values
and checking whether the strong target is met.  Every candidate header ``protects''
all collected weak headers (note that all of these weak headers point to the same previous
strong header).

In order to keep the number of found weak headers close to a constant value,
\name adjusts the difficulty $T_w$ of weak headers every 2016 blocks immediately
following the adjustment of the difficulty $T_s$ of the strong headers according
to \autoref{eqn:adjust_Ts}, such that the ratio $T_w / T_s$ is kept at a
constant (we discuss its value in \autoref{sec:analysis}).

\subsection{Block Layout and Validation}
\label{sec:details:block:layout_valid}
A block in our scheme consists of transactions, a list of weak headers, and a
strong header that authenticates these transactions and weak headers.  Strong
and weak headers in our system inherit the fields from Bitcoin headers and
additionally enrich it by a new field.  A block header consists of the following
fields:
\begin{compactdesc}
	\item[$\mathit{PrevHash}$\textnormal{:}] is a hash of the previous block header, 
	\item[$\mathit{Target}$\textnormal{:}] is the value encoding the current target defining
        the difficulty of finding new blocks,
	\item[$\mathit{Nonce}$\textnormal{:}] is a nonce, used to generate PoW,
	\item[$\mathit{Timestamp}$\textnormal{:}] is a Unix timestamp,
    \item[$\mathit{TxRoot}$\textnormal{:}] is the root of the Merkle tree~\cite{Merkle1988}
        aggregating all transactions of the block, and
    \item[$\mathit{Coinbase}$\textnormal{:}] represents an address of the miner that will
        receive a reward. 
\end{compactdesc}
As our protocol rewards finders of weak headers (see details in
\autoref{sec:details:payments}), every weak header has to be accompanied with the
information necessary to identify its finder.  Otherwise, a finder of a strong
block could maliciously claim that some (or all) weak headers were found by her
and get rewards for them.  For this purpose and for efficiency, we introduced a
new 20B-long header field named \textit{Coinbase}.  With the
introduction of this field, \name headers are 100B long.  But on the other
hand, there is no longer any need for Bitcoin coinbase transactions (see
\autoref{sec:pre:bitcoin}), as all rewards are determined from headers.

In our scheme, weak headers are exchanged among nodes as part of a block, hence
it is necessary to protect the integrity of all weak headers associated with the
block. 
To realize it, we introduce a
special transaction, called a \textit{binding transaction}, which contains a hash
value computed over the weak headers.  This transaction is the first transaction
of each block and it protects the collected weak headers. 
Whenever a strong header is found, it is announced
together with all its transactions and collected weak headers, therefore, this
field  protects all associated weak headers. To encode this field we utilize
the \texttt{OP\_RETURN} operation as follows:
\begin{equation}
    \label{eq:weak_hdr_hash}
    \texttt{OP\_RETURN}\quad \mathit{H(hdr_0 \| hdr_1 \| ... \| hdr_n)},
\end{equation}
where $hdr_i$ is a weak header pointing to the previous strong header.  
Since  weak headers have redundant fields (the \mbox{\textit{PrevHash}}, \textit{Target}, and
\textit{Version} fields have the same values as the strong header), 
we propose to save bandwidth and storage by not including these fields into the data of a block. 
This modification reduces the size of a weak header from 100B to 60B only, which is especially important for SPV clients who keep downloading new block headers.

With our approach, a newly mined and announced block can encompass
multiple weak headers.  Weak headers, in contrast to strong headers, are not used to
authenticate transactions, and they are even stored and exchanged \textit{without} their
corresponding transactions. Instead, the main purpose of including weak headers
it to contribute and reflect the aggregated mining power concentrated on a given
branch of the blockchain.  We present a fragment of a blockchain of \name in
\autoref{fig:block}.  As depicted in the figure, each block contains a single strong
header, transactions, and a set of weak headers aggregated via a binding
transaction.

\begin{figure}[t!]
  \centering
  \includegraphics[width=\columnwidth]{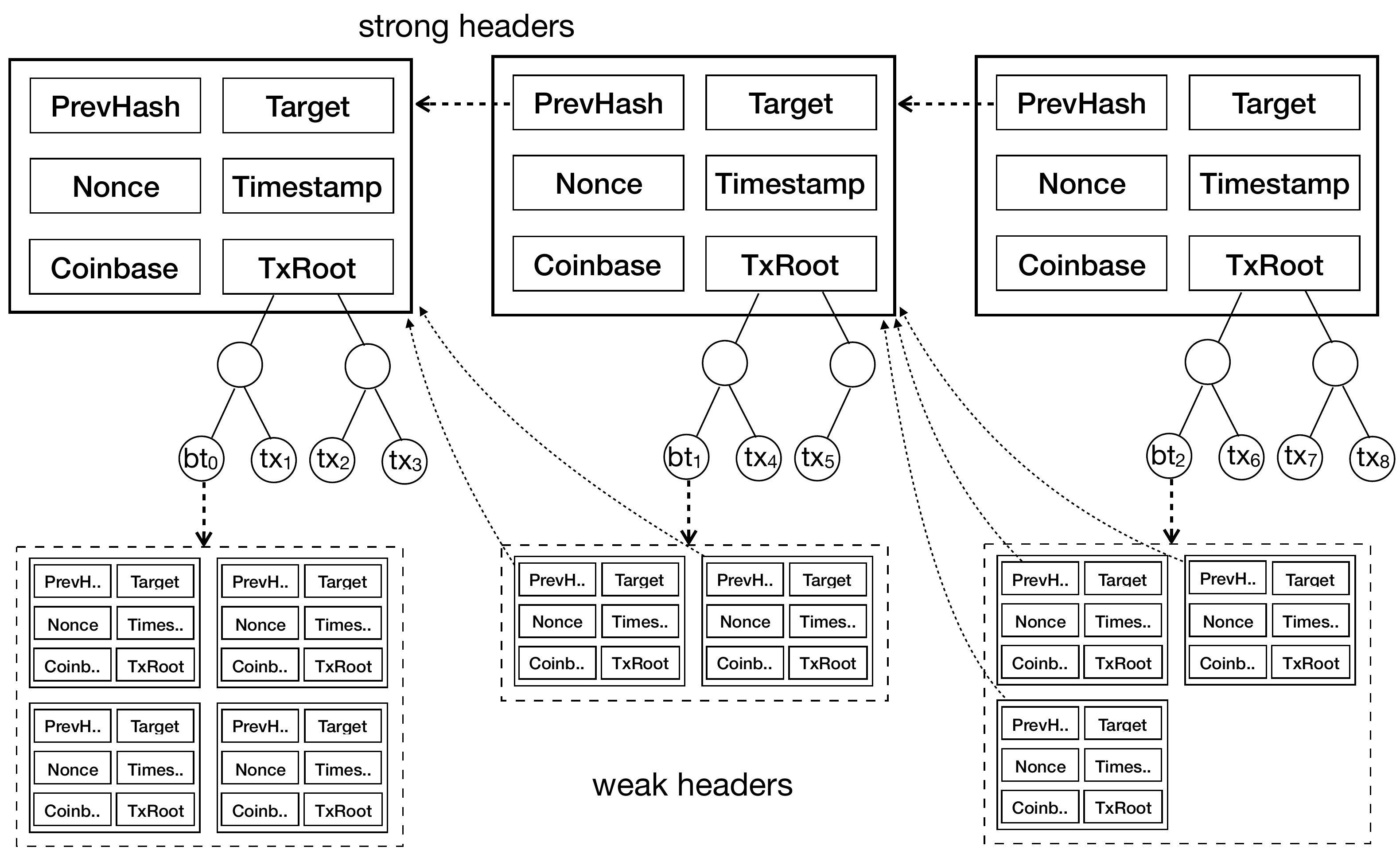}
    \caption{An example of a blockchain fragment with strong headers, weak
    headers, and binding and regular transactions.}
  \label{fig:block}
\end{figure}

On receiving a new block, miners validate the block by checking the following
(see \textit{validateBlock()} in \autoref{alg:all}):
\begin{compactenum}
    \item The strong header is protected by the PoW and points to the previous
        strong header.
    \item Header fields have correct values (i.e., the version, target, and
        timestamp are set correctly).
    \item All included transactions are correct and protected by the strong
        header. This check also includes checking that all weak headers
        collected are protected by a binding transaction included in the block.

    \item All included weak headers are correct:
        a) they meet the targets as specified in \autoref{eq:weak},
        b) their \textit{PrevHash} fields point to the
                previous strong header,
        and c) their version, targets, and timestamps have correct values.
\end{compactenum}
If the validation is successful, the block is accepted as part of the
blockchain.

\subsection{Forks}
\label{sec:details:forks}
One of the main advantages of our approach is that blocks reflect their aggregated mining
power  more precisely.  Each block beside its strong header contains
multiple weak headers that contribute to the block's PoW.  In the case of a
fork, our scheme relies on the strongest chain rule, however, the PoW is
computed differently than in Bitcoin.  For every chain its PoW is calculated as
presented by the \textit{chainPoW()} procedure in \autoref{alg:all}.  Every
chain is parsed and for each of its blocks the PoW is calculated by adding:
\begin{compactenum}
    \item the PoW of the strong header, computed as $T_{max}/T_s$, where $T_{max}$ is the maximum target value, and
    \item the accumulated PoW of all associated weak headers,
        counting each weak header equally as $T_{max}/T_w$.
\end{compactenum}
Then the chain's PoW is expressed as just the sum of all its blocks' PoW. Such
an aggregated chain's PoW is compared with the competing chain(s). The chain with
the largest aggregated PoW is determined as the current one.  As difficulty in
our protocol changes over time, the strong target $T_s$ and PoW of weak headers
are relative to the maximum target value $T_{max}$.  We assume that nodes of the
network check whether every difficulty window is computed correctly (we skipped
this check in our algorithms for easy description).

Including and empowering weak headers in our protocol moves away from
Bitcoin's ``binary'' granularity and gives blocks better expression of the PoW
they convey. An example is presented in \autoref{fig:chain}. For instance, nodes
having the blocks $B_i$ and $B'_i$ can immediately decide to follow the
block $B_i$ as it has more weak headers associated, thus it has accumulated more
PoW than the block $B'_i$.

An exception to this rule is when miners solve conflicts.  Namely, on
receiving a new block, miners run the algorithm as presented, however, they also
take into consideration PoW contributions of known weak headers that
point to the last blocks. For instance, for a one-block-long fork within the
same difficulty window, if a block $B$ includes $l$ weak headers and a miner
knows of $k$ weak headers pointing to $B$, then that miner will select $B$ over any competing block $B'$ that includes $l'$ weak and has $k'$ known weak headers pointing to it if $l+k>l'+k'$.  Note that this rule incentivizes miners to propagate
their solutions as quickly as possible as competing blocks become ``stronger'' over
time.

\begin{figure}[t!]
  \centering
  \includegraphics[width=.9\linewidth]{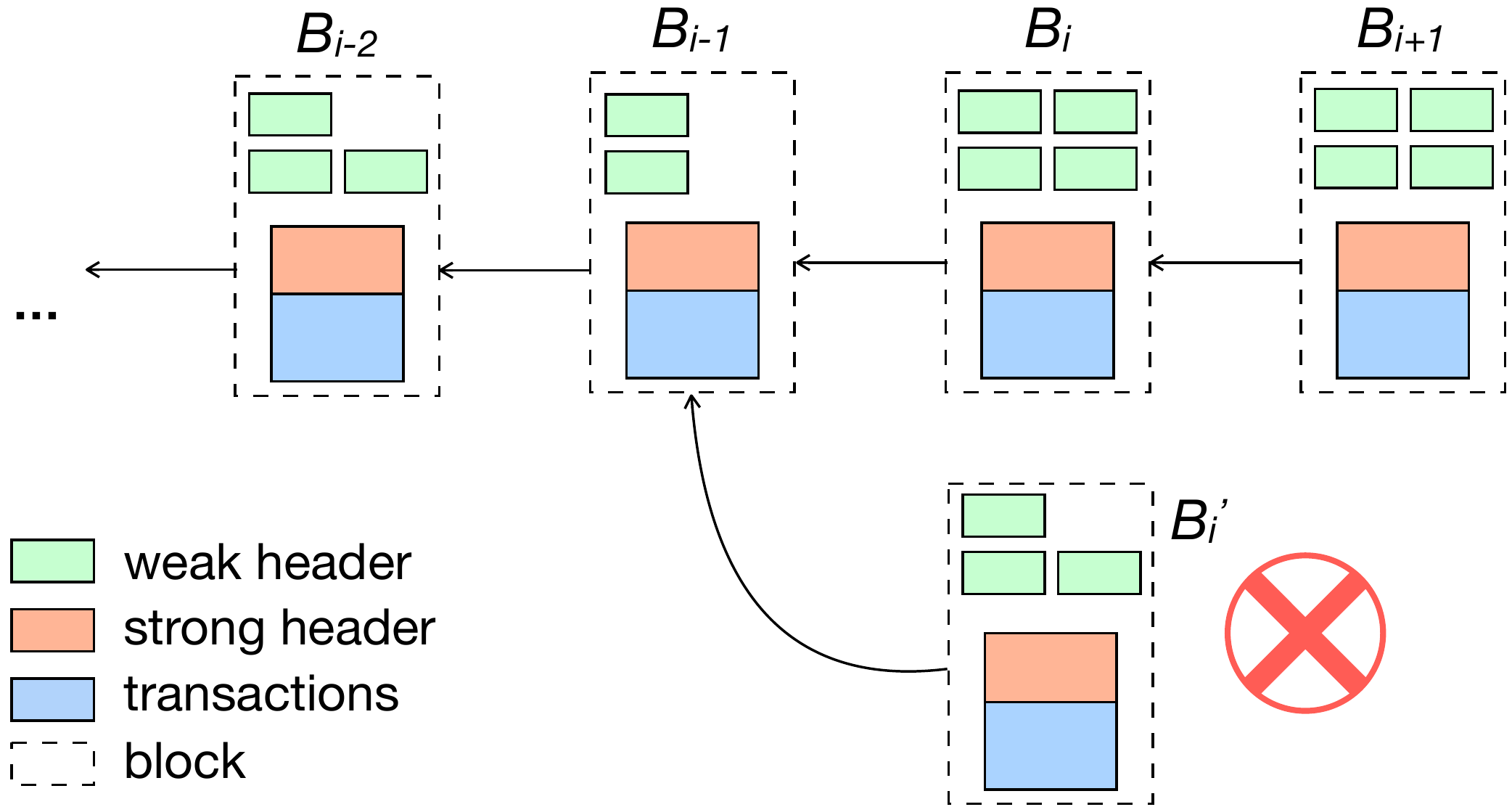}
    \caption{An example of a forked blockchain in \name.}
  \label{fig:chain}
\end{figure}

\subsection{Rewarding Scheme}
\label{sec:details:payments}
The rewards distribution is another crucial aspect of \name and it is presented
by the \textit{rewardBlock()} procedure from \autoref{alg:all}.  The miner that
found the strong header receives the full reward $R$. Moreover,  in contrast to
Bitcoin, where only the ``lucky'' miner is paid the full reward, in our scheme
all miners that have contributed to the block's PoW (i.e., whose weak headers
are included) are paid by commensurate rewards to the provided PoW.    A weak
header finder receive a fraction of $R$, i.e., $\gamma*c*R*T_s/T_w$, as a reward
for its corresponding solution contributing to the total PoW of a particular
branch, where the $\gamma$ parameter influences the relative impact of weak header rewards and $c$ is just a scaling constant (we discuss
their potential values and implications in \autoref{sec:analysis}).  Moreover, we
do not limit weak header rewards and miners can get multiple rewards for their
weak headers within a single block.  Similar reward mechanisms are present in
today's mining pools (see \autoref{sec:related}), but unlike them, weak header
rewards in \name are independent of each other.  Therefore, the reward scheme is
not a zero-sum game and miners cannot increase their own rewards by dropping
weak headers of others (actually, as we discuss in \autoref{sec:analysis}, they
can only lose since their potential solutions would have less PoW without
others' weak headers).  Furthermore, weak header rewards decrease significantly
the mining variance as miners can get steady revenue, making the system more
decentralized and collaborative.

As mentioned before, the number of weak headers of a block is unlimited, they
are rewarded independently (i.e., do not share any reward), and all block
rewards in our system are proportional to the PoW contributed.  In such a
setting, a mechanism incentivizing miners to terminate a block creation is
needed (without such a mechanism, miners could keep creating huge blocks with weak
headers only).  In order to achieve this,  \name always attributes block
transaction fees ($B.TxFees$)  to the finder of the strong header (who also
receives the full reward $R$).

Note that in our rewarding scheme, the amount of newly minted coins is always at
least $R$, and consequently, unlike Bitcoin or Ethereum~\cite{wood2014ethereum},
the total supply of the currency in our protocol is not upper-bounded. This
design decision is made in accordance with recent results on the long-term
instability of deflationary
cryptocurrencies~\cite{longrunBitcoin,instability_noReward,gapgame}.

\subsection{Timestamps}
\label{sec:details:timestamps}
In \name, we follow the Bitcoin rules on constraining timestamps (see
\autoref{sec:pre:bitcoin}), however, we redefine how block timestamps are
interpreted.  Instead of solely relying on a timestamp put by the miner who
mined the block, block timestamps in our system are derived from the strong
header and all weak headers included in the corresponding block. The algorithm
to derive a block's timestamp is presented as \textit{getTimestamp()} in
\autoref{alg:all}.  A block's timestamp is determined as a weighted average
timestamp over the strong header's timestamp and all timestamps of the weak
headers included in the block.  The strong header's timestamp has a weight of
$1$, while weights of weak header timestamps are determined as their PoW
contributed (namely, a weak header's timestamp has a weight of the ratio between
the strong target and the weak target).  Therefore, the timestamp value is
adjusted proportionally to the mining power associated with a given block.  That
change reflects an average time of the block creation and mitigates miners that
intentionally or misconfigured put incorrect timestamps into the blockchain. We
show the effectiveness of this approach in \autoref{sec:analysis:ts}.

\subsection{SPV Clients}
Our protocol supports light SPV clients. With every new block, an SPV client is
updated with the following information:
\begin{equation}
    \mathit{hdr, hdr_0, hdr_1, ..., hdr_n, BTproof},
\end{equation}
where \textit{hdr} is a strong header, $\mathit{hdr_i}$ are associated weak headers, and \textit{BTproof} is an inclusion proof of a binding transaction
that contains a hash over the weak headers (see \autoref{eq:weak_hdr_hash}).
Note that headers contain redundant fields, thus as described in
\autoref{sec:details:block:layout_valid}, they can be provided to SPV clients
efficiently.

With this data, the client verifies fields of all headers, computes the PoW of
the block (analogous, as in \textit{chainPoW()} from \autoref{alg:all}), and
validates the \textit{BTproof} proof to check whether all weak headers are
correct, and whether the transaction is part of the blockchain (the proof is
validated against \textit{TxRoot} of \textit{hdr}).  Afterward, the client saves the strong header
$\textit{hdr}$ and its computed PoW, while other messages (the weak headers and
the proof) can be dropped.

%% file: sec/analysis.tex
In this section, we evaluate the requirements discussed in \autoref{sec:requirements}.
We start with analyzing \name's efficiency and practicality. Next, we study how
our design helps with reward variance, chain quality, and security.

%

\subsection{Efficiency and Practicality}

For the efficiency, it is important to consider the main source of additional load on the bandwidth, storage, and processing power of the nodes: the weak headers. 
Hence, in the following section we analyze the probability distribution of the number of weak headers. Next, we discuss the value of the impact of the parametrization on the average block rewards.

\subsubsection{Number of Weak Headers}
\newcommand{\totalhp}{\eta}

In Bitcoin, we assume that hashes are drawn randomly between 0 and $T_{max} = 2^{256}-1$. Hence, a single hash being smaller than $T_w$ is a \emph{Bernoulli trial} with parameter \mbox{$p_w = T_w/2^{256}$}. The number of hashes tried until a weak 
header is found is therefore \emph{geometrically} distributed, and the time in 
seconds between two weak headers is approximately \emph{exponentially} distributed 
with rate $\totalhp p_w$, where $\totalhp$ is the total hash rate per second and $p_w$ is chosen such that $\totalhp p_w \approx 1/600$. When a weak header is found, it is also a strong block 
with probability $p_s/p_w$ (where $p_s = T_s/2^{256}$), which is again a Bernoulli trial. Hence, the probability distribution of the number of weak headers 
found between two strong blocks is that of the number of trials before the first 
successful trial --- as such, it also follows a geometric distribution, but
with mean $p_w / p_s - 1$.\footnote{Another way to reach this conclusion is as follows: 
the number of weak headers found in a fixed time interval is Poisson distributed, 
and it can be shown that the number of Poisson arrivals in an interval with 
exponentially distributed length is geometrically distributed.} For example, 
for $T_w / T_s = 2^{10}$ this means that the average number of weak headers 
per block equals 1023. With 60 bytes per weak header (see \autoref{sec:details:block:layout_valid}) and 1MB per Bitcoin block, this would mean 
that the load increases by little over 6\% on average with a small computational
overhead introduced (see details in
\autoref{sec:implementation}). The probability of 
having more than 16667 headers (or 1MB) in a block would equal.\footnotemark
\footnotetext{For an actual block implementation, we advice to introduce
separate spaces for weak headers and transactions. With such a design, miners do
not have incentives and trade-offs between including more transactions instead
of weak headers.}
$$
\left(1 - \frac{p_s}{p_w}\right)^{16668} = \left(1-2^{-10}\right)^{16668} \approx 8.4603 \cdot 10^{-8}.
$$
Since around 51,000 Bitcoin blocks are found per year, this is expected to happen roughly once every 230 years.

\subsubsection{Total Rewards}\label{sec:scaling}


To ease the comparison to the Bitcoin protocol, 
we can enforce the same average mining reward per block (currently 12.5 BTC).
Let $R$ denote Bitcoin's mining reward. 
Since we reward weak headers as well as strong blocks, we need to scale all mining 
rewards by a constant $c$ to ensure that the total reward remains unchanged --- this is done in the \textit{rewardBlock} function in \autoref{alg:all}.
As argued previously, we reward all weak headers equally by $\gamma R  T_s/T_w $. Since the average number of weak headers per strong block is \mbox{$T_w/T_s - 1$}, this means that the expected total reward per block (i.e., strong block and weak header rewards) equals
$cR + cR\gamma T_s/T_w \cdot (T_w/T_s - 1)$.
Hence, we find that
$$
c = \frac{1}{1+\gamma(T_w/T_s - 1) T_s/T_w},
$$
which for large values of $T_w/T_s$ is close to $1/(1+\gamma)$. This means that if $\gamma = 1$, the strong block and weak header rewards contribute almost equally to a miner's total reward.


\renewcommand{\P}{\mathbb{P}}
\newcommand{\E}{\mathbb{E}}
\newcommand{\N}{\mathbb{N}}
\newcommand{\var}{\textnormal{Var}}
\renewcommand{\i}{{\bf 1}}

\subsection{Reward Variance of Solo Mining}
The tendency towards centralization in Bitcoin caused by powerful mining pools
can largely be attributed to the high reward variance of solo mining~\cite{rosenfeld2011analysis,gencer2018decentralization}.
Therefore, keeping the reward variance of a solo miner at a low level is a central design goal.

Let $\mathbf{R}^{BC}$ and $\mathbf{R}^{SC}$ be the random variables representing the per-block
rewards for an $\alpha$-strong solo miner 
in Bitcoin and in \name, respectively.
For any given strong block in both protocols, we define the random variable $\mathbf{I}$ as follows:
\begin{equation*}
\mathbf{I} = \begin{cases}
1 &\text{the block is mined by the solo miner,}\\
0 &\text{otherwise.}
\end{cases}
\end{equation*}
By definition, $\mathbf{I}$ has a Bernoulli distribution, which means that $\E(\mathbf{I}) = \alpha$ and $\var(\mathbf{I}) = \alpha(1 - \alpha)$,
where $\E$ and $\var$ are the mean and variance of a random variable respectively. The following technical lemma will aid our analysis of the reward variances of solo miners:
\begin{lemma} \label{lm:randomsum}
Let $X_1,X_2,\ldots$ be independent and identically distributed random variables. Let N be defined on $\{0,1,\ldots\}$ and independent of $X_1,X_2,\ldots$. Let $N$ and all $X_i$ have finite mean and variance. Then
\begin{equation*}
\var\left(\sum_{i=1}^N X_i \right) = \E(N)\var(X) + \var(N)(\E(X))^2. \label{eq:randomsum}
\end{equation*}
\label{lm:randomsum}
\end{lemma}
\begin{proof}
See \cite{randomsums}.
\end{proof}
\paragraph{Reward Variance of Solo Mining in Bitcoin.}
Bitcoin rewards the miner of 
a block creator with the fixed block reward $R$ and the variable (total)
mining fees, which we denote by the random variable $\textbf{F}$. Therefore, we have
\begin{equation*}
\mathbf{R}^{BC} = \mathbf{I}(R + \textbf{F}),
\end{equation*}
which implies that
\begin{equation}\label{eq:bitcoinreward1}
\begin{split}
\var(\mathbf{R}^{BC}) & = R^2 \var(\mathbf{I}) +  \var(\mathbf{I} \textbf{F}).
\end{split}
\end{equation}
Since $ \mathbf{I} \textbf{F} = \sum_{i=0}^{\mathbf{I}} \textbf{F}$, we can use Lemma~\ref{lm:randomsum}
(substituting $\textbf{I}$ for $N$ and $\textbf{F}$ for $X$) to obtain
\begin{equation}\label{eq:bitcoinreward2}
\begin{split}
\var\left(\mathbf{I} \textbf{F} \right) &= \E(\mathbf{I})\var(\textbf{F}) + \var(\mathbf{I})\E^2(\textbf{F}).
\end{split}
\end{equation}
Combining \eqref{eq:bitcoinreward1} and \eqref{eq:bitcoinreward2} gives
\begin{equation}\label{eq:bitcoinreward3}
\begin{split}
\var(\mathbf{R}^{BC}) & = \E(\mathbf{I})\var(\textbf{F}) + \var(\mathbf{I}) \left( \E^2(\textbf{F}) + R^2\right) \\
& = \alpha \var(\textbf{F}) +\alpha(1-\alpha) \left(\E^2(\textbf{F}) + R^2\right). \\
\end{split}
\end{equation}
When the fees are small compared to the mining reward, 
this simplifies to $\alpha(1-\alpha)R^2$. By comparison, in \cite{rosenfeld2011analysis} the variance of the block rewards (without fees) earned by a solo miner across a time period of $t$ seconds is studied, and found to equal $\alpha R^2 t / 600$.\footnote{In particular, it is found to be $htR^2/(2^{32}D)$, where $h = \alpha \eta$ and $\eta / (2^{32} D) \approx 1/600$.} The same quantity can be obtained by using \eqref{eq:bitcoinreward3}, Lemma~\ref{lm:randomsum}, and the total number of strong blocks found (by any miner) after $t$ seconds of mining (which has a Poisson distribution with mean $t/600$). 
\paragraph{Reward Variance of Solo Mining in StrongChain.}
For $\mathbf{R}^{SC}$, we assume that the solo miner has $\mathbf{N}$
weak headers included in the strong block, and that she obtains $c \gamma R T_s/T_w$ reward per weak header.
Then the variance equals
$$
\mathbf{R}^{SC} = \mathbf{I}(cR + \mathbf{F}) + c \gamma R T_s/T_w \mathbf{N},
$$
where $c$ is the scaling constant derived in \autoref{sec:scaling}.
Hence, by applying Lemma~\ref{lm:randomsum},
we compute the variance of $\mathbf{R}^{SC}$ as
\begin{equation}
\begin{split}
\var(\mathbf{R}^{SC}) = & \; (cR)^2\var(\mathbf{I}) + \var(\mathbf{I} \mathbf{F}) \\  & \; + (c\gamma RT_s/T_w)^2\var(\mathbf{N}).
\end{split}
\label{eq:threeterms}
\end{equation}
The first term, which represents 
the variance of the strong block rewards, is similar to Bitcoin but 
multiplied by $c^2$. If we choose $T_w/T_s = 1024$ and $\gamma = 10$ (this choice is motivated later in this section), 
$c^2$ roughly equals $0.0083$, which is quite small. Hence, the strong block rewards 
have a much smaller impact on the reward variance in our setting than in Bitcoin. 
The second term, which represents the variance of the fees, 
is precisely the same as for Bitcoin. 
The third term represents the variance of the weak header rewards, which in turn completely depends on $\var(\mathbf{N})$. 

To evaluate $\var(\mathbf{N})$, we again use Lemma~\ref{lm:randomsum}: let, for any weak header, $\mathbf{J}$ equal $1$ if it is found by the solo miner, and $0$ otherwise. Also, let $\mathbf{L}$ be the total number of weak headers found in the block, so including those not found by the solo miner. Then $\mathbf{N}$ is the sum of $\mathbf{L}$ instances of $\mathbf{J}$, where $\mathbf{J}$ has a Bernoulli distribution with success probability $\alpha$ (and therefore $\E(\mathbf{J}) = \alpha$ and $\var(\mathbf{J}) = \alpha(1-\alpha)$), and $\mathbf{L}$ has a geometric distribution with success probability $T_s/T_w$ (and therefore $\E(\mathbf{L}) = T_w/T_s-1$ and $\var(\mathbf{L}) = (T_w/T_s)^2 - T_w/T_s$. By substituting this into \eqref{eq:threeterms}, we obtain:
\begin{equation}
\begin{split}
\var(\mathbf{N}) = & \; \E(\mathbf{L})\var(\mathbf{J}) + \var(\mathbf{L})(\E(\mathbf{J}))^2 \\
= & \; (T_w/T_s-1) \alpha (1-\alpha) \\
& \; + ((T_w/T_s)^2 - T_w/T_s ) \alpha^2
\end{split}
\label{eq:varn}
\end{equation}
Substituting \eqref{eq:varn} for $\var(\mathbf{N})$ and $\alpha(1-\alpha)$ for $\var(\mathbf{I})$ into \eqref{eq:threeterms} then yields an expression that can be evaluated for different values of $T_w/T_s$, $\gamma$, and $\alpha$, as we discuss in the following.
\paragraph{Comparison}

The difference between between \eqref{eq:bitcoinreward3} and \eqref{eq:threeterms} in practice is illustrated in \autoref{fig:variances}. This is done by comparing for a range of different values of $\alpha$ the block rewards' \emph{coefficient of variation}, which is the ratio of the square root of the variance to the mean. 

To empirically validate the results, we have also implemented a simulator in Java that can evaluate Bitcoin as well as \name. We use two nodes, one of which controls a share $\alpha$ of the hash rate, and another controls a share $1-\alpha$. The nodes can broadcast information about blocks, although we abstract away from most of the other network behavior. We do not consider transactions (i.e., we mine empty blocks), and we use a simplified model for the propagation delays: delays are drawn from a Weibull distribution with shape parameter $0.6$ \cite{papagiannaki2003measurement}, although for \autoref{fig:variances} the mean was chosen to be negligible  (more realistic values are chosen for \autoref{tab:latencies}).

\begin{figure}[!t]
    \centering
\includegraphics[width=0.4\textwidth, trim={2cm 0 1.2cm 0}, clip]{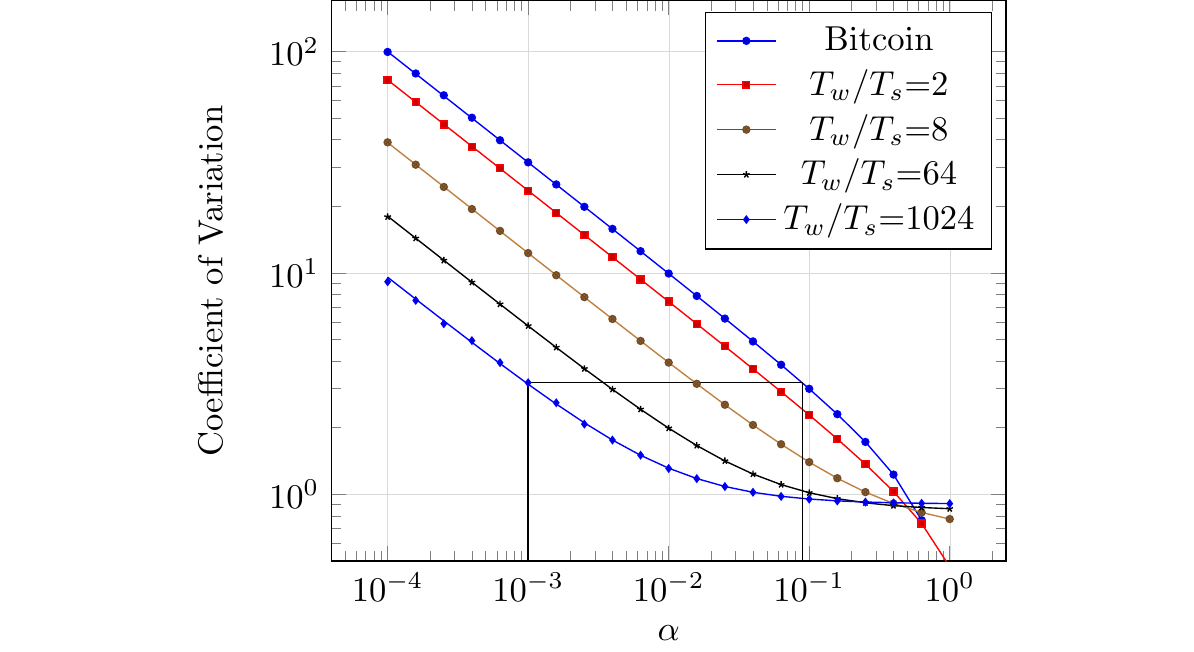}
\caption{\emph{Coefficients of variation} for the total rewards of $\alpha$-strong miners for 
	different strong/weak header difficulty ratios ($T_w/T_s = 1$ corresponds to Bitcoin). The lines indicate the exact results obtained using our analysis, whereas the markers indicate simulation results.
	We used \mbox{$\gamma = \log_2(T_w/T_s)$}. The black lines indicate that for \mbox{$T_w/T_s = 1024$}, a $0.1\%$-strong miner has a coefficient of variation that is comparable to a $9\%$-strong miner's in Bitcoin.
	}
\label{fig:variances}
\end{figure}

The black lines in \autoref{fig:variances} demonstrate that when $T_w/T_s=1024$, 
	a miner with share $0.1\%$ of the mining power has the same coefficient of reward variation as a miner with stake $9\%$ in Bitcoin. 
	Also note that for $T_w/T_s=1024$ and $\alpha \ge 1\%$,
	the coefficient of variation does not substantially decrease 
	anymore, because nearly all of the reward variance is due to the number of weak headers. Hence, there would be fewer reasons for miners in our system 
to join large and cooperative mining pools, which has a positive effect on the decentralization of the system.

%

\subsection{Chain Quality}
\label{sec:analysis:chainq}

One measure for the `quality' of a blockchain is the stale rate of blocks \cite{gervais2016security}, i.e., the percentage of blocks that appear during forks and do not make it onto the main chain. This is closely related to the notion of mining power utilization~\cite{eyal2016bitcoinng}, which is the fraction of mining power used for non-stale blocks.
In \name, the stale rate of strong blocks may increase due to high latency.
After all, while a new block is being propagated through the network, weak headers that strengthen the previous block that are found will be included by miners in their PoW calculation. As a result, some miners may refuse to switch to the new block when it arrives. However, the probability of this happening is very low: because each weak header only contributes $T_s/T_w$ to the difficulty of a block, it would take on average 10 minutes to find enough weak headers to outweigh a block. As we can see in \autoref{tab:latencies}, the effect on the stale rate is negligible even for very high network latencies (i.e., 53 seconds).
We also emphasize that the strong block stale rate is less important in our
setting, as the losing miner still would benefit from her weak headers appended
to the winning block. 

 Regarding the fairness, defined as the ratio between the observed share of the rewards (we simulate using one $10\%$-strong miner and a $90\%$-strong one) and the share of the mining power, we see that \name does slightly worse than Bitcoin for high network latencies. The most likely cause is that due to the delay in the network, the $10\%$-strong miner keeps mining on a chain that has already been extended for longer than necessary. This gives the miner a slight disadvantage compared to the $90\%$-strong miner.

\begin{table*}[!h]
	\centering
    \small
	\begin{tabular}{rcccccccccc}
		\toprule 
		& \multirow{4}*{\textbf{Latency}} & \multirow{4}*{\textbf{Bitcoin}} &  \multicolumn{8}{c}{\textbf{\name}} \\
		\expandafter\cline\expandafter{\expandafter4\string-11\smallskip}
		
		\multicolumn{3}{c}{} & \multicolumn{1}{c}{$\mathbf{T_w/T_s = 2}$} & \multicolumn{3}{c}{$\mathbf{T_w/T_s = 64}$}  & & \multicolumn{3}{c}{$\mathbf{T_w/T_s = 1024}$}  \\%
		\expandafter\cline\expandafter{\expandafter5\string-7}
		\expandafter\cline\expandafter{\expandafter9\string-11\smallskip}
		&  & & $\gamma=1$ & $\gamma = 1$ & $\gamma=7$ & $\gamma=63$ & & $\gamma=1$ & $\gamma=10$ & $\gamma = 1023$ \\
		\midrule


   & low & .0023 & .0025 & .0021 & .0026 & .0028 & & .0023 & .0025 & .0019 \\
   strong stale rate & medium & .0073 & .0082 & .0087 & .0077 & .0078 & & .0084 & .0067 & .0081 \\
   & high & .0243 & .0297 & .0242 & .0263 & .0247 & & .0274 & .0249 & .0263 \\
  \midrule
   & low & --- & .0043 & .0047 & .0049 & .0046 & & .0049 & .0047 & .0047 \\
   weak stale rate & medium & --- & .0142 & .0151 & .0154 & .0149 & & .0145 & .0147 & .0149 \\
   & high & --- & .0400 & .0459 & .0474 & .0452 & & .0469 & .0455 & .0463 \\
  \midrule
   & low & .9966 & .9814 & .9749 & .9747 & .9838 & & .9645 & .9809 & .9812 \\
   fairness & medium & .9276 & .9384 & .9570 & .9360 & .9364 & & .9329 & .9400 & .9385 \\
   & high & .7951 & .7640 & .7978 & .7820 & .7757 & & .7756 & .7766 & .7775 \\
		
		\bottomrule
	\end{tabular}
    \caption{For several different protocols, the strong block stale rate, weak header rate, and the `fairness' for an $\alpha$-strong honest miner with $\alpha=0.1$. Here, fairness is defined as the ratio between the observed share of the reward and the `fair' share of the rewards (i.e, 0.1). 'Low', 'medium', and 'high' latencies refer to the mean of the delay distribution  in the simulator; these are roughly 0.53 seconds, 5.3 seconds, and 53 seconds respectively. The simulations are based on a time period corresponding to roughly 20\,000 blocks.}
\label{tab:latencies}
\end{table*}

\subsection{Security}
\label{sec:analysis:security}

One of the main advantages of \name is the added robustness to selfish mining strategies akin to those discussed in \cite{eyal2014majority} and \cite{sapirshtein2016optimal}. In selfish mining, attackers aim to increase their share of the earned rewards by tricking other nodes into mining on top of a block that is unlikely to make it onto the main chain, thus wasting their mining power. This may come at a short-term cost, as the chance of the attacker's blocks going stale is increased --- however, the difficulty rescale that occurs every 2016 blocks means that if the losses to the honest nodes are structural, the difficulty will go down and the gains of the attacker will increase.

In the following, we will consider the selfish mining strategy of \cite{eyal2014majority},\footnote{The `stubborn mining' strategy of \cite{sapirshtein2016optimal} offers mild improvements over \cite{eyal2014majority} for powerful miners, but the comparison with \name is similar. We have also modeled \name using a Markov decision process, in a way that is similar to the recently proposed framework of~\cite{zhang2019lay}. Due to the state space explosion problem, we could only investigate the protocol with a small number of expected weak headers, but we have not found any strategies noticeably that are better than those presented.} described as follows:
\setdefaultleftmargin{1em}{2em}{}{}{}{}
\begin{compactitem}
\item The attacker does not propagate a newly found block until she finds at least a second block on top of it, and then only if the difference in difficulty between her chain and the strongest known alternative chain is between zero and $R$.
\item The attacker adopts the strongest known alternative chain if its difficulty is at least greater than her own by $R$.
\end{compactitem}
\setdefaultleftmargin{2em}{2em}{}{}{}{}
In Figure~\autoref{fig:selfish_payoffs1}, we have depicted the profitability of this selfish mining strategy for different choices of $T_w/T_s$. As we can see, for $T_w/T_s = 1024$ the probability of being `ahead' after two strong blocks is so low that the strategy only begins to pay off when the attackers' mining power share is close to $45\%$ --- this is an improvement over Bitcoin, where the threshold is closer to $33\%$. 

\begin{figure}[t!]
    \centering
    \subfloat[][]{\includegraphics[width=0.235\textwidth, trim={2cm 0 1.7cm 0}, clip]{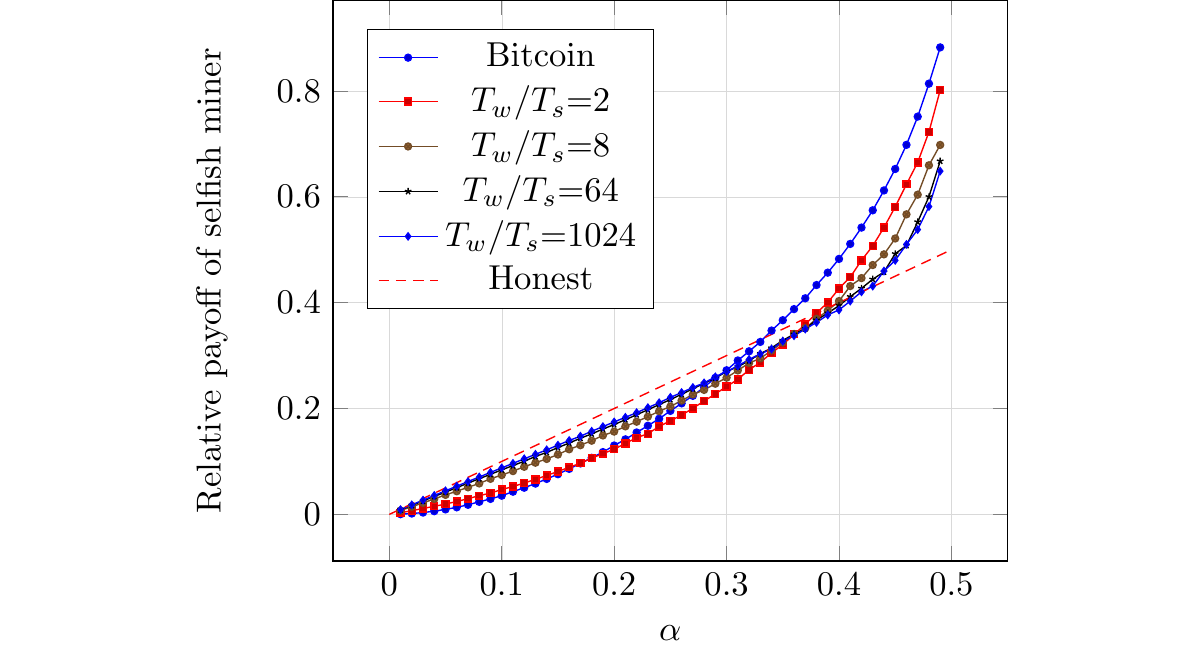}\label{fig:selfish_payoffs1}}
     \subfloat[][]{\includegraphics[width=0.235\textwidth, trim={2cm 0 1.7cm 0}, clip]{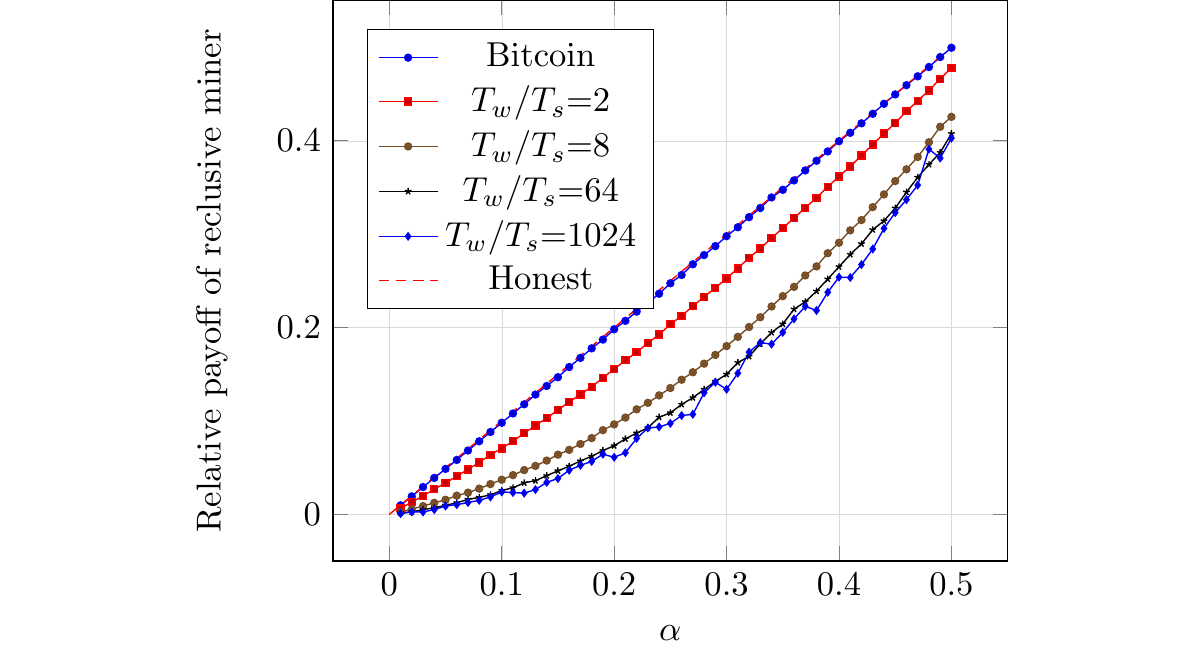}\label{fig:selfish_payoffs2}}\\
         \subfloat[][]{\includegraphics[width=0.235\textwidth, trim={2cm 0 1.7cm 0}, clip]{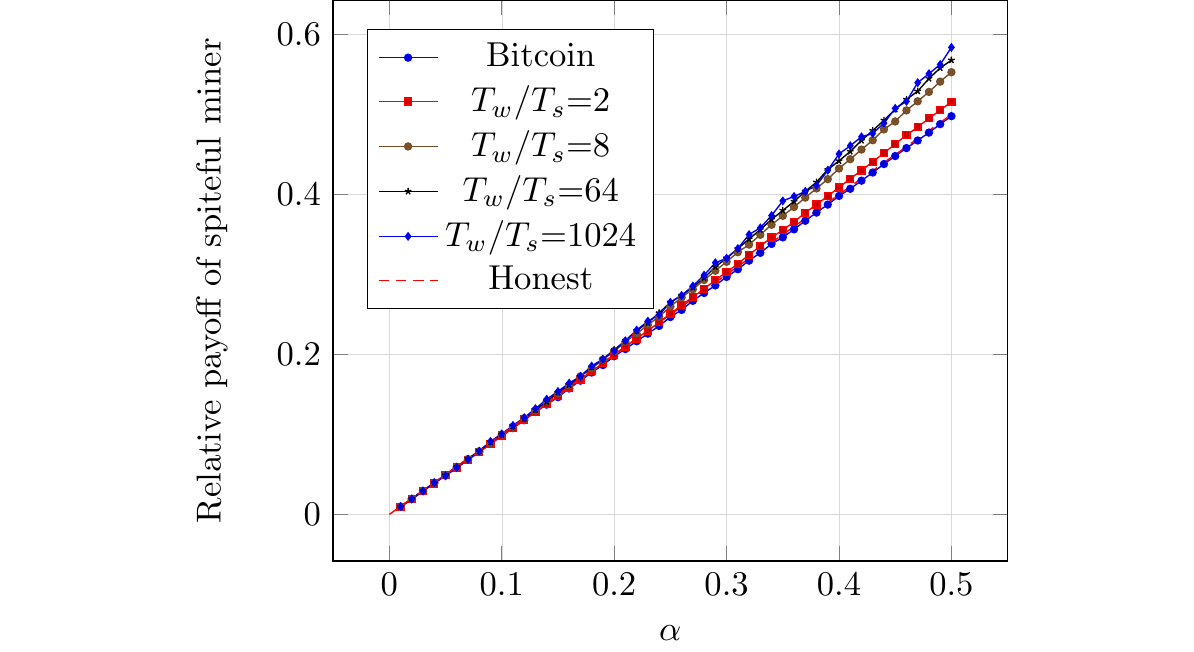}\label{fig:selfish_payoffs3}}
     \subfloat[][]{\includegraphics[width=0.235\textwidth, trim={2cm 0 1.7cm 0}, clip]{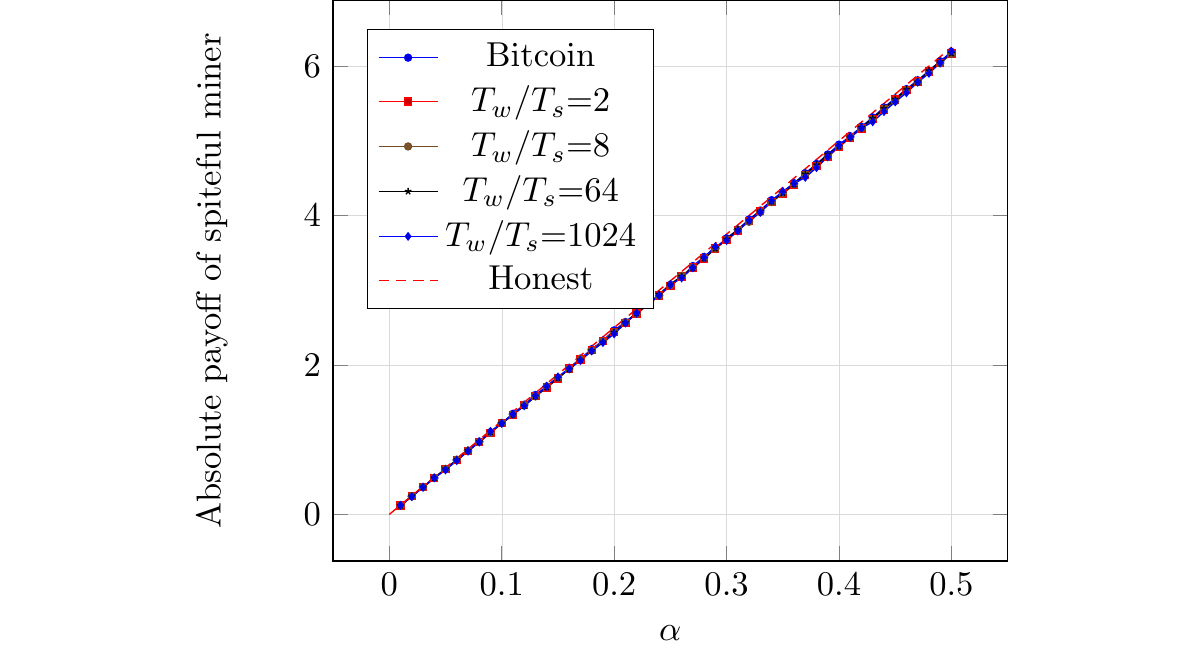}\label{fig:selfish_payoffs4}}
	\caption{Payoffs of an $\alpha$-strong adversarial miner for different strategies. Figure (a): relative payoff of a \emph{selfish} miner following the strategy of \cite{eyal2014majority}, compared to an $(1-\alpha)$-strong honest miner. Figure (b): relative payoff of a \emph{reclusive} miner who does not broadcast her weak blocks. Figure (c): \emph{relative} payoff (with respect to the rewards of all miners combined) of a \emph{spiteful} miner, who does not include other miners' weak blocks unless necessary. Figure (d): \emph{absolute} payoff of a \emph{spiteful} miner, with 12.5 BTC on average awarded per block. We consider Bitcoin and \name with different choices of $T_w/T_s$, with \mbox{$\gamma = \log_2(T_w/T_s)$}.
}
	
\end{figure}

\name does introduce new adversarial strategies based on the mining of new weak headers.  
Some examples include not broadcasting any newly found weak blocks (``reclusive'' mining), refusing to include the weak headers of other miners (``spiteful'' mining), and postponing the publication of a new strong block and wasting the weak headers found by other miners in the meantime. In the former case, the attacker risks losing their weak blocks, whereas in both of the latter two cases, the attacker risks their strong block going stale as other blocks and weak headers are found. Hence, these are not cost-free strategies. Furthermore, because the number of weak headers does not affect the difficulty rescale, the attacker's motive for increasing the stale rate of other miners' weak headers is less obvious (although in the long run, an adversarial miner could push other miners out of the market entirely, thus affecting the difficulty rescale). 

In Figure~\ref{fig:selfish_payoffs2}, we have displayed the relative payout (with respect to the total rewards) of a reclusive $\alpha$-strong miner --- this strategy does not pay for any $\alpha<0.5$. In Figure~\ref{fig:selfish_payoffs3}, we have depicted the relative payoff of a spiteful mine who does not include other miners' weak blocks unless necessary (i.e., unless others' weak blocks together contribute more than $R$ to the difficulty, which would mean that any single block found by the spiteful miner would always go stale). For low latencies (the graphs were generated with an average latency of 0.53 seconds), the strategy is almost risk-free, and the attacker does manage to hurt other miners more than herself, leading to an increased relative payout. However, as displayed in Figure~\ref{fig:selfish_payoffs4}, there are no absolute gains, even mild losses. As mentioned earlier, the weak headers do not affect the difficulty rescale so there is no short-term incentive to engage in this behavior --- additionally there is little gain in computational overhead as the attacker still needs to process her own weak headers. In Section~\ref{sec:parametrization} we will discuss protocol updates that can mitigate these strategies regardless.

\subsection{More Reliable Timestamps}
\label{sec:analysis:ts}
Finally, we conducted a series of simulations to investigate how the introduced
redefinition of timestamps interpretation (see \textit{getTimestamp()} in
\autoref{alg:all} and \autoref{sec:details:timestamps})   influences the
timestamp reliability in an adversarial setting.  We assume that an adversary
wants to deviate blockchain timestamps by as much as possible.  There are two
strategies for such an attack, i.e., an adversary can either ``slow down''
timestamps or ``accelerate'' them.  In the former attack, the best adversary's
strategy is to use the minimum acceptable
timestamp in every header created by the adversary.  Namely, the adversary sets
its timestamps to the median value of the last eleven blocks (a header with a
lower timestamp would not be accepted by the network -- see
\autoref{sec:problem:issues}).  As for the latter attack, the
adversary can analogously bias timestamps towards the future by putting the maximum acceptable value in all her created headers.  The maximum timestamp value accepted by network nodes
is two hours in the future with respect to the nodes' internal clocks (any header with a higher
timestamp would be rejected).

\begin{figure}[t!]
    \centering
    \subfloat[][]{\includegraphics[width=0.255\textwidth, trim={2cm 0 1.7cm 0}, clip]{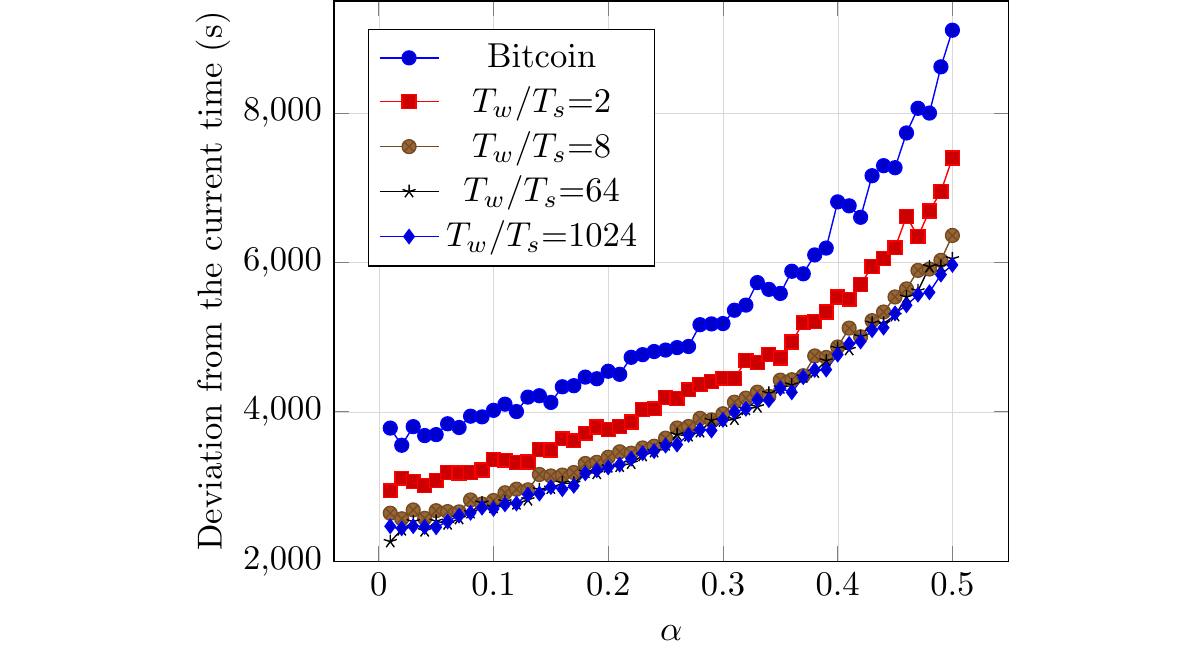}}
    \subfloat[][]{\includegraphics[width=0.215\textwidth, trim={3.3cm 0 1.7cm 0}, clip]{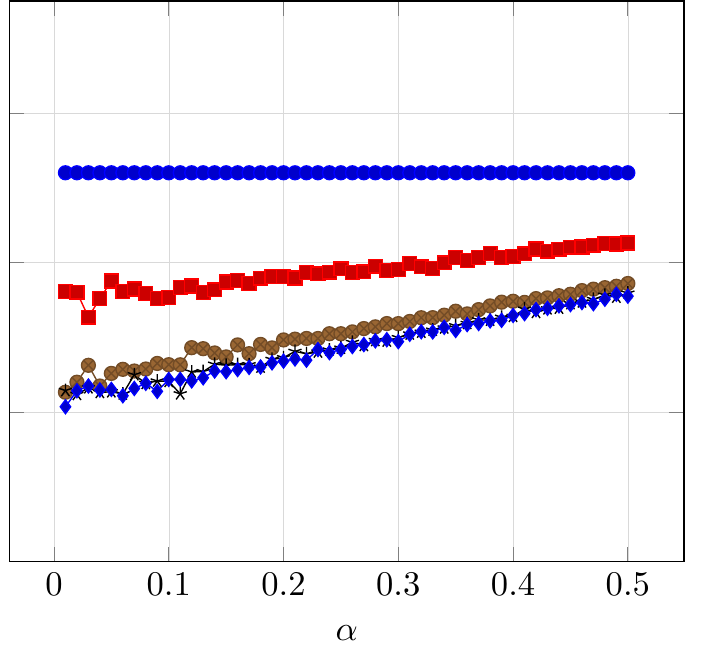}}
	\caption{The deviation from the network time that an $\alpha$-strong adversary 
    can introduce for its mined blocks by slowing (the left graph) and
    accelerating (the right graph) timestamps.}
	\label{fig:timestamps}
\end{figure}
In our study, we assume that honest nodes maintain the network time which the
adversary tries to deviate from.  We consider the worst-case scenario, which is when the
adversary, who also biases all her header timestamps, mines the strong block. We
measure (over 10000 runs) how such a malicious timestamp can be mitigated by our
redefinition of the block timestamps interpretation.  We present the obtained
results in \autoref{fig:timestamps}, and as shown in the slow-down case our
protocol achieves much more precise timestamps than Bitcoin (the difference is
around 2000 seconds).  Similarly, when the adversary accelerates timestamps, our
protocol can mitigate it effectively, adjusting the adversarial timestamps by
2000-3500 seconds towards the correct time.  This effect is achieved due to the
block's timestamp calculation as a weighted average of all block headers.
The adversary could try to remove honest participants' weak headers in order to
give a stronger weight to its malicious timestamps, but in
\autoref{sec:parametrization} we discuss ways to mitigate it.

%% file: sec/appendix.tex
\subsection{Impact of the Parameter Choice}
\label{sec:parametrization}

The results presented in \autoref{sec:analysis} required several parameters to be fixed. First of all, we had to choose $\gamma$, which determines the relative contribution of the weak headers to the total mining rewards. Second, there is the contribution of the weak blocks to the chain difficulty, which in the \mbox{\emph{chainPoW()}} function in \autoref{alg:all} was set to be only $T_{max} / T_w$. This means that the PoW of a weak header relative to a strong block's PoW --- we call this the \textit{difficulty factor} --- is fixed to be $T_{s} / T_{w}$. In the following, we first discuss the relevant trade-offs and then motivate our choice.

When both $\gamma$ and the difficulty factor are low, the impact on the reward variance of the miners (as per \autoref{fig:variances}) will be mild as the strong block rewards still constitute about 50\% of the mining rewards. This reliance on the block rewards also means that `spiteful' mining as discussed in \autoref{sec:analysis:security} is disincentivized as the risk of strong blocks going stale still has a considerable impact on total rewards. However, selfish mining as proposed in \cite{eyal2014majority} relies on several blocks in a row being mined in secret, and even for a low difficulty factor it becomes much harder for the attacker's chain to stay `ahead' of the honest chain, as the latter accumulates strength from the weak headers at a faster rate. Hence, in this setting we only gain protection against selfish mining.

When $\gamma$ is high but the difficulty factor is not (which is the setting of \autoref{sec:analysis}), then in addition to disincentivizing selfish mining, the reward variances become much less dependent on the irregular strong block rewards. This benefits small miners and reduces centralization, as we also discuss in \autoref{sec:centralized_pools}. However, spiteful mining will have more of an impact as the possible downside (i.e., a latency-dependent increase in the strong block stale rate) will have less of an effect on the total rewards. 

When both $\gamma$ and the difficulty factor are high, the impact of spiteful mining is mitigated. The reason is that blocks quickly accumulate enough weak headers to outweigh a strong block, and in this case spiteful miners need to adopt the other weak blocks or risk their strong block becoming stale with certainty. The downside in this setting is that the system-wide block stale rate is increased. For example, if each weak header contributes $\gamma T_s / T_w$ to the difficulty and $\gamma = 10$, then after (on average) one minute enough weak headers are found to outweigh a strong block, and if propagation of the block takes longer than one minute then some miners will not adopt the block, increasing the likelihood of a fork.

In this paper, we have chosen the second of the three approaches --- a moderately high $\gamma$, yet a minor difficulty factor. The reason is that the only downside (spiteful mining) was considered less of a concern than the other downsides (namely a low impact on reward variances and a higher block stale rate respectively) for two reasons: a) because spiteful mining does not lead to clear gains for the attacker, and b) because it only has a large impact on other miners' profits if the attacker controls a large share of the mining power, whereas the emergence of large mining pools is exactly what \name discourages. The specific value of $\gamma =10$ for $T_w/T_s = 1024$ (or $\gamma = \log_2 (T_w/T_s)$ in general) was chosen to sufficiently reduce mining reward variances, yet leaving some incentive to discourage spiteful mining.

The protocol can be further extended to disincentivize spiteful mining, e.g., by
additionally awarding strong block finders a reward that is proportional to the number of weak
headers included. This would make \name more similar to Ethereum, where stale block (`uncle') rewards are paid both to the miner of a stale block and the miner of the successful block that included it (see \autoref{sec:related} for additional discussion of Ethereum's protocol). However, we leave such modifications and their
consequences as future work. 

\subsection{\name and Centralized Mining}
\label{sec:centralized_pools}
Decentralized mining pools aim to reduce variance while providing benefits for
the system (i.e., trust minimization for pools, and a higher number of validating nodes).
However, mining in Bitcoin is in fact dominated by centralized mining pools
whose value proposition, over decentralized pools, is an easy setup and
participation.  Therefore, rational miners motivated by their own benefit,
instead of joining decentralized pools prefer centralized ``plug-and-play''
mining.  It is still debatable whether centralized mining pools are beneficial or
harmful to the system. However, it has been proved multiple times, that the
concentration of significant computing power caused by centralized mining is
risky and should be avoided, as such a strong pool has multiple ways of
misbehaving and becomes a single point of failure in the system. One example is the pool GHash.IO, which in 2014 achieved more than 51\% of the mining power. This undermined trust in the Bitcoin network to the extent that the pool was forced to actively ask miners to join other pools~\cite{ghash}.

In order to follow incentives of rational miners, \name does not require any
radical changes from them and is compatible with
centralized mining pools; however, it is specifically designed to mitigate their
main security risk (i.e., power centralization). In \name such pools could be
much smaller than in Bitcoin (due to minimized variance) and to support this
argument we conducted a study.  We listed the largest Bitcoin mining pools and
their shares in the global mining power (according to
\url{https://www.blockchain.com/en/pools} as for the time of writing). Then for
each pool, we calculated what would be the pool size in \name to offer the miner the same payout variance experience, and the variance reduction factor in that case.  As shown in \autoref{tab:pools}, for the Bitcoin largest mining pool with 18.1\% of the global hash rate, an equivalent pool in \name (to provide miners the same reward experience) could be as small as 0.245\% of the hash rate -- around 74 times smaller.  Even better reduction factors are achieved for smaller pools.  Therefore, our study indicates that \name makes the size of a pool almost an irrelevant factor for miners' benefits (i.e., there is no objective advantage of joining a large pool over a medium or a small one).  Therefore we envision that with \name, centralized mining pools will naturally be much more distributed.

\begin{table}[tb!]
	\centering
	\begin{tabular}{lrrc}
		\toprule 
        \multirow{2}*{Mining Pool}            &  \multicolumn{2}{c}{Pool Size} & Size \\
        &  Bitcoin &   \name & Reduction\\
		\midrule
        BTC.com & 18.1\% & 0.245\%    & 74$\times$ \\
        F2Pool & 14.1\% & 0.172\%     & 82$\times$ \\
        AntPool & 11.7\% & 0.135\%    & 87$\times$ \\
        SlushPool & 9.1\% & 0.099\%   & 92$\times$ \\
        ViaBTC & 7.5\% & 0.079\%      & 95$\times$ \\
        BTC.TOP & 7.1\% & 0.074\%     & 96$\times$ \\
        BitClub & 3.1\% & 0.030\%     & 103$\times$ \\
        DPOOL & 2.6\% & 0.025\%       & 104$\times$ \\
        Bitcoin.com & 1.9\% & 0.018\% & 106$\times$ \\
        BitFury & 1.7\% & 0.016\%     & 106$\times$ \\
		\bottomrule
	\end{tabular}
    \caption{Largest Bitcoin mining pools and the corresponding pool sizes in
    StrongChain offering the same relative reward variance (\mbox{$T_w/T_s=1024$} and
    $\gamma=10$).}
    \label{tab:pools}
\end{table}

\subsubsection*{Limitations}
As discussed, it is beneficial for the system if as many participants as possible independently run full nodes; however, miners
join large centralized pools not only due to high reward variance.
Other potential reasons include
the minimization of operational expenses as running a full node is a large overhead,
higher efficiency since large pools may use high-performance hardware and network,
better ability to earn extra income from merge mining~\cite{narayanan2016bitcoin},
better protection against various attacks, anonymity benefits, etc.
This work focuses on removing the reward variance reason.  Although we believe that \name would
produce a larger number of small pools in a natural way,
it does not eliminate the other reasons, so some large
centralized pools may still remain.  Luckily, our system is orthogonal to multiple concurrent
solutions.  For instance, \name could be easily combined with non-outsourceable
puzzle schemes (see \autoref{sec:related}) to increase the number of full nodes
by explicitly disincentivizing miners from outsourcing their computing power. We
leave such a combination as interesting future work.

%% file: sec/implementation.tex
We implemented our system in order to investigate its feasibility and confirm
the stated properties.  We implemented a \name full node with interactive client in
Python, and our implementation includes the complete logic from \autoref{alg:all}
and all functionalities required to have a fully operational system
(communication modules, message types, validation logic, etc...).\footnote{Our
implementation is available at
\url{https://github.com/ivan-homoliak-sutd/strongchain-demo/}.}
As described
before, the main changes in our implementation to the Bitcoin's block layout
are:
\begin{compactitem}  
    \item a new (20B-long) \textit{Coinbase} header field, 
    \item a new binding transaction protecting all weak headers of the block, 
    \item removed original coinbase transaction,
\end{compactitem}  
where a binding transaction has a single (32B-long) output as presented in
\autoref{eq:weak_hdr_hash}.\footnote{An alternative choice is to store a hash of
weak headers in a header itself. Although simpler, that option would incur a
higher overhead if the number of weak headers is greater than several.}

Weak headers introduced by our system  impact the bandwidth and storage overhead
(when compared with Bitcoin).  Due to compressing them
(see \autoref{sec:details:block:layout_valid}), the size of a single weak header in a
block is 60B.  For example, with an average number of weak headers equal 1024,
the storage and bandwidth overhead increases by about 61.5KB per block (e.g.,
with 64 weak headers, the overhead is only 3.8KB).  Taking into account the
average Bitcoin block size of about 1MB (the average between 15 Oct and 15 Nov
2018\footnote{\url{https://www.blockchain.com/en/charts/avg-block-size}}), 1024
weak headers constitute around 6.1\% of today's blocks, while 64 headers only
0.4\%.  The same overhead is introduced to SPV clients, that besides a strong
header need to obtain weak headers and a proof for their corresponding binding
transaction. Thus, an SPV update (every 10 minutes) would be 61.5KB or 3.8KB on
average for 1024 or 64 weak headers, respectively. However, since only strong
headers authenticate transactions, SPV clients do not need to store weak headers
and after they are validated, they can remove them (they need to just calculate
and associate their aggregated PoW with the strong header).  Such an approach
would not introduce any noticeable storage overhead on SPV clients.

Nodes validate all incoming weak headers; however, this overhead is a single hash
computation and simple sanity checks per header. Even with our unoptimized
implementation running on a commodity PC the total validation of a single weak header takes
around 50$\mu$s on average (i.e., 51ms per 1024 headers on a single core). 
Given that we do not believe this overhead can lead to more serious
denial-of-service attacks than ones already known and existing in Bitcoin (e.g.,
spamming with large invalid blocks).  Additionally, \name can adopt prevention
techniques present in Bitcoin, like blacklisting misbehaving peers.

%% file: sec/related.tex
Employing weak solutions (and their variations) in Bitcoin is an idea~\cite{wBlockgoodandbad,wBlockSimulator}
circulating on Bitcoin forums for many years.  Initial proposals leverage weak
solutions (i.e., \textit{weak blocks}) for faster transaction
confirmations~\cite{PT1,nearBlockFirst}, for signaling the current working
branch of particular miners~\cite{GA1,GA2,TN1}.  Unfortunately, most of these
proposals come without necessary details or lack rigorous analysis.  Below, we
discuss the most related attempts that have been made to utilize weak or stale blocks in
PoW-based decentralized consensus protocols. 
We compare these systems in \autoref{tab:comparison} according to their reward
and PoW calculation schemes.

\myparagraph{Subchains.}
Rizun proposes Subchains~\cite{rizun2016subchains}, where a
chain of weak blocks (a so-called subchain) bridging each
pair of subsequent strong blocks is created. The design of Subchain puts a special focus on
increasing the transaction throughput and the double-spend security for 
unconfirmed transactions. Rizun argues that since the (weak) block interval of subchains
is much smaller than the strong block interval, it allows for faster (weak) transaction
confirmations. Another claimed advantage of such an approach is that during the process of building subchains, the miners can
detect forks earlier, and take actions accordingly to avoid wasting computational power. 
However, the design of Subchain sidesteps a concrete security analysis
at the subchain level.
In detail, by using a chaining data structure where
one weak header referencing the previous weak header in a subchain, 
it introduces high stale rate on a subchain.
More importantly, due to applying a Bitcoin-like
subchain selection policy in case of conflicts,
this approach is vulnerable to the selfish mining attack launched on a
subchain.

\begin{table*}[t!]
	\centering
    \small
	\begin{tabular}{rccccc}
		\toprule 
		~ & Bitcoin v0.1 & Bitcoin  & Fruitchains & Flux  &\name \\
		\midrule
		\rowcolor{black!20} Reward (strong) & $R+F$ & $R+F$ & 0 & $(R+F)/(E+1)$ & $cR + F$\\
		Reward (weak) & 0 & 0 &   $(R+F)/E$ & $(R+F)/(E+1)$ & $c\gamma R T_s / T_{w} $ \\
		\rowcolor{black!20} Chain weight contrib. (strong) & 1 & $T_{max}/T_s$ & $T_{max}/T_s$ & $T_{max}/T_s$ &  $T_{max}/T_s$ \\
		Chain weight contrib. (weak) &  0 & 0 & 0 & 0 & $T_{max}/T_w$\\
		\bottomrule
	\end{tabular}
    \caption{The comparison of reward and PoW computation schemes of \name and the related
    systems. ($F$: block transaction fees, $E$: expected number of weak headers
    per block. The entries for Flux are approximations for the PPLNS scheme in
    P2Pool, on which it is based.)}
    \label{tab:comparison}
\end{table*}

\myparagraph{Flux.}
Based on similar ideas as Subchain, Zamyatin et al. propose Flux~\cite{zamyatinflux}.
In contrast to Subchain, Flux shares rewards (from newly minted coins and transaction fees)
evenly among the finders of weak and strong blocks according to the computational resources
they invested. This approach reduces the reward variance of miners, and therefore mitigates the need
for large mining pools, which is beneficial for the system's decentralization. 
In addition, simulation experiments show that Flux renders selfish mining on the main chain less profitable.
However, alike Subchains, Flux employs a chain structure for weak blocks, which
inevitably introduces race conditions, increasing the stale rate of weak blocks and making it 
more susceptible to selfish mining attacks at the subchain level.
The designers of Flux let both of these issues open and 
discuss the potential application of GHOST~\cite{2013-Sompolinsky} to subchains. 
Another limitation of this work is that the authors do not analyze 
the requirements on space consumption when putting possibly 
a high number of overlapping transactions into Flux subchains, 
which could negatively influence 
network, storage, and processing resources.

\noindent\emph{Remarks on Subchain and Flux.~}
One important difference between our approach and
the above two designs is that we adopt a flat hierarchy for the weak
blocks, which not only eliminates the possibility of selfish mining in a set of weak solutions,
but also avoids the issue of stale rate of weak blocks.
In contrast, both Subchain and Flux employ a chain structure for weak blocks, 
inevitably making them more susceptible to selfish mining 
attacks at the subchain level.
Moreover, in our approach rewards are not shared, therefore miners
can only benefit from appending received weak solutions.
In addition, none of Subchain and Flux provide a concrete
implementation demonstrating their applicability.

\myparagraph{FruitChains.}
Another approach to address the mining variance and selfish mining issues
is the FruitChains protocol proposed by Pass and Shi~\cite{pass2017fruitchains}.
In FruitChains, instead of directly storing the records inside blocks, the
records or transactions are put inside ``fruits'' with relatively low mining difficulties. Fruits then are appended to a blockchain via blocks
which are mined with a higher difficulty. Mined fruits and blocks yield rewards, hence, miners can be paid more often and there is no need 
to form a mining pool.

However, some practical and technical details of FruitChains lead to undesired
side effects.  First, the scheme allows fruits with small difficulties to be
announced and accepted by other miners. With too small difficulty it could
render high transmission overheads or even potential denial-of-service attacks
and its effects on the network are not investigated. On the other hand, too high
fruit difficulty could result in a low transaction throughput and a high reward
variance. Second, duplicate fruits are discarded, even though they might be
found by different miners -- this naturally implies some stale fruit rate
(uninvestigated in the paper).  Similarly, it is unclear would a block finder
have an incentive to treat all fruits equally and to not prioritize her mined
fruits (especially when fruits are associated with transaction fees).  Moreover,
fruits that are not appended to the blockchain quickly enough have to be mined
and broadcast again, rendering additional overheads.  Finally, the description
of FruitChains lacks important details (e.g., the size of the fruits or the
overheads introduced) as well as an actual implementation.

\myparagraph{GHOST and Ethereum.}
An alternative approach for decreasing a high reward variance of miners is to shorten the block creation rate to the extent that does not
hurt the overall system security -- such an approach increases transaction throughput as well.
The Greedy Heaviest-Observed Sub-Tree (GHOST) chain selection rule~\cite{2013-Sompolinsky}
makes use of stale blocks to increase the weight of their ancestors, which
achieves a 600 fold speedup for the block generation compared to Bitcoin, 
while preserving its security strength.
Despite the inclusion of stale blocks in the blockchain,
only the miners of the main chain get rewards for the inclusion of the stale blocks.

In contrast to the original GHOST, 
the white and yellow papers of Ethereum~\cite{ethereum-white-paper,wood2014ethereum} 
propose to reward also miners of stale blocks in order to further increase the security -- 
not wasting with the consumed resources for mining of stale blocks. 
However, Ritz and Zugenmaier shows that rewarding so called ``uncle blocks'' lowers
the threshold at which selfish mining is profitable~\cite{ritz2018impact} --
a selfish miner can built-up the ``heaviest'' chain, 
as she can reference blocks previously not broadcast to the honest network.
Likewise, the inclusive blockchain protocol~\cite{lewenberg2015inclusive}, which
increases the transaction throughput,
but leaves the selfish mining issue unsolved.

\myparagraph{DAG-based Protocols.}
SPECTRE~\cite{sompolinsky2016spectre} is an example of the protocols family that
leverages a directed acyclic graph (DAG). This family proposed more radical
design changes motivated by the observation that one essential throughput
limitation of Nakamoto consensus is the data structure leveraged which can
be maintained only sequentially.  SPECTRE generalizes the Nakamoto's blockchain
into a DAG of blocks, while allowing miners to add blocks concurrently with a
high frequency, just basing on their individual current views of the DAG.  Such
a design provides multiple advantages over chain-based protocols including
\name.  Frequently published blocks increase transaction throughput and provide
fast confirmation times while relaxed consistency requirements allow to tolerate
propagation delays.  Like \name, SPECTRE also aims to decrease mining reward
variance, but achieves it again via frequent blocks.  However, frequent blocks
have a side effect of transaction redundancy which negatively impacts the storage
and transmission overheads, and this aspect was not investigated.  Moreover,
SPECTRE provides a weaker property than chain-based consensus protocols as
simultaneously added transactions cannot be ordered.  This and schemes following
a similar design are payments oriented and to support order-specific
applications, like smart contracts, they need to be enhanced with an additional
ordering logic.

More recently, Sompolinsky and Zohar~\cite{cryptoeprint:2018:104} proposed two
DAG-based protocols improving the prior work.  PHANTOM introduces and uses a
greedy algorithm (called the GHOSTDAG protocol) to determine the order of
transactions. This eliminates the applicability issues of SPECTRE, but for the
cost of slowing down transaction confirmation times.  Combining advantages of
PHANTOM and SPECTRE into a full system was left by the authors as a future work.

\myparagraph{Decentralization-oriented Schemes.}
Mining decentralization was a goal of multiple previous proposals.  One
direction is to design mining such that miners do
not have incentive to outsource resources or forming coalitions.  Permacoin~\cite{miller2014permacoin}
is an early attempt to achieve it where miners instead of proving their work
prove that their store (fragments of) a globally-agreed file.  Permacoin is
designed such that: a) payment private keys are bound to puzzle solutions
-- outsourcing private keys is risky for miners, b) sequential and random storage
access is critical for the mining efficiency, thus it disincentives miners from
outsourcing data.  If the file is valuable, then a side-effect of Permacoin is
its usefulness, as miners replicate the file. 

The notion of non-outsourceable mining was further extended and other schemes
were proposed~\cite{miller2015nonoutsourceable,zeng2017nonoutsourceable}.
Miller et al.~\cite{miller2015nonoutsourceable} introduces ``strongly
non-outsourceable puzzles'' that aim to disincentivize pool creation by
requiring all pool participants to remain honest. In short, with these puzzles
any pool participant can steal the pool reward without revealing its identity.
The scheme relies on zero knowledge proofs requiring a trusted setup and
introducing significant computational overheads.  The scheme is orthogonal to
\name and could be integrated with easily integrated with \name, however, after
a few years of their introduction no
system of this class was actually deployed at scale; thus, we do not have any empirical
results confirming their promised benefits.

SmartPool is a different approach that was proposed by Luu et al.~\cite{203704}.
In SmartPool, the functionality of mining pools was implemented as a smart
contract.  Such an approach runs natively only on smart-contract platforms but it
allows to eliminate actual mining pools and their managers (note that SmartPool
still imposes fees for running smart contracts), while
preserving most benefits of pool mining.

\myparagraph{Rewarding Schemes for Mining Pools.}
Mining pools divide the block reward (together with the transaction fees) in such a way
that each miner joining the pool is paid his fair share in proportion to his
contribution.  Typically, the contribution of an individual miner in the pool is
witnessed by showing weak solutions called {\it shares}. 

There are various rewarding schemes that mining pools employ.  The
simplest and most natural method is the proportional scheme where the reward of
a strong block is divided in proportion to the numbers of shares
submitted by the miners. However, this scheme leads to pool hopping
attacks~\cite{poolhop}. To avoid this security issue, many other rewarding
systems are developed, including the Pay-per-last-N-shares (PPLNS) scheme and
its variants.  We refer the reader to \cite{rosenfeld2011analysis} for a
systematic analysis of different pool rewarding systems.
 
The reward mechanisms in \name can be seen conceptually as a mining pool
built-in into the protocol.  However, there are important differences between our
design and the mining pools.  The most contrasting one is that in \name
rewarding is not a zero-sum game and miners do not share rewards.  In mining
pools, all rewards are shared and this characteristic  causes multiple in- and
cross-pool attacks that cannot be launched against our scheme.  Furthermore, the
miner collaboration within Bitcoin mining pools is a ``necessary evil'', while
in \name the collaboration is beneficial for miners and the overall system.
We discuss \name and mining pools further in \autoref{sec:centralized_pools}.

%% file: sec/conclusions.tex
In this paper, we proposed a transparent and collaborative proof-of-work
protocol.  Our approach is based on Nakamoto consensus and Bitcoin,
however, we modified their core designs.  In particular, in contrast to them,
we take advantage of weak solutions, which although they do not finalize a block
creation positively contribute to the blockchain properties.  We also
proposed a rewarding scheme such that miners can benefit from exchanging and
appending weak solutions. These modifications lead to a more secure, fair, and
efficient system. Surprisingly, we show that our approach helps with
seemingly unrelated issues like the freshness property.  Finally, our
implementation indicates the efficiency and deployability of our approach.

Incentives-oriented analysis of consensus protocols is a
relatively new topic and in the future we would like to extend our work by
modeling our protocol with novel frameworks and tools.
Another topic worth investigating in future is how to combine \name with systems
solving other drawbacks of Nakamoto
consensus~\cite{eyal2016bitcoinng,kogias2016enhancing,luu2016secure}, or how to mimic the protocol in the proof-of-stake setting.

%% file: paper.bbl
\begin{thebibliography}{10}

\bibitem{bahack2013theoretical}
L.~Bahack.
\newblock Theoretical {B}itcoin attacks with less than half of the
  computational power (draft).
\newblock {\em arXiv preprint arXiv:1312.7013}, 2013.

\bibitem{bayer1993improving}
D.~Bayer, S.~Haber, and W.~S. Stornetta.
\newblock Improving the efficiency and reliability of digital time-stamping.
\newblock In {\em Sequences II}. Springer, 1993.

\bibitem{culubas}
A.~Boverman.
\newblock Timejacking \& {B}itcoin.
\newblock
  \url{https://culubas.blogspot.sg/2011/05/timejacking-bitcoin_802.html}, 2011.

\bibitem{instability_noReward}
M.~Carlsten, H.~A. Kalodner, S.~M. Weinberg, and A.~Narayanan.
\newblock On the instability of {B}itcoin without the block reward.
\newblock In {\em Proceedings of the 2016 {ACM} {SIGSAC} Conference on Computer
  and Communications Security}, 2016.

\bibitem{courtois2014subversive}
N.~T. Courtois and L.~Bahack.
\newblock On subversive miner strategies and block withholding attack in
  {B}itcoin digital currency.
\newblock {\em arXiv preprint arXiv:1402.1718}, 2014.

\bibitem{douceur2002sybil}
J.~R. Douceur.
\newblock The {S}ybil attack.
\newblock In {\em International workshop on peer-to-peer systems}. Springer,
  2002.

\bibitem{randomsums}
S.~Dunbar.
\newblock Random sums of random variables.
\newblock
  \url{http://www.math.unl.edu/~sdunbar1/ProbabilityTheory/Lessons/Conditionals/RandomSums/randsum.shtml}.

\bibitem{dwork1992pricing}
C.~Dwork and M.~Naor.
\newblock Pricing via processing or combatting junk mail.
\newblock In {\em Annual International Cryptology Conference}. Springer, 1992.

\bibitem{eyal2015miner}
I.~Eyal.
\newblock The miner's dilemma.
\newblock In {\em 2015 IEEE Symposium on Security and Privacy (SP)}. IEEE,
  2015.

\bibitem{eyal2016bitcoinng}
I.~Eyal, A.~E. Gencer, E.~G. Sirer, and R.~Van~Renesse.
\newblock Bitcoin-{NG}: A scalable blockchain protocol.
\newblock In {\em Proceedings of NSDI}, 2016.

\bibitem{eyal2014majority}
I.~Eyal and E.~G. Sirer.
\newblock Majority is not enough: {B}itcoin mining is vulnerable.
\newblock In {\em International conference on financial cryptography and data
  security}. Springer, 2014.

\bibitem{ghash}
C.~Farivar.
\newblock {Bitcoin pool GHash.io commits to 40\% hashrate limit after its 51\%
  breach}.
\newblock
  \url{https://arstechnica.com/information-technology/2014/07/bitcoin-pool-ghash-io-commits-to-40-hashrate-limit-after-its-51-breach/},
  2014.

\bibitem{GA1}
{Gavin Andresen}.
\newblock {Faster blocks vs bigger blocks}.
\newblock
  \url{https://bitcointalk.org/index.php?topic=673415.msg7658481#msg7658481},
  2014.

\bibitem{GA2}
{Gavin Andresen}.
\newblock {Weak block thoughts}.
\newblock
  \url{https://lists.linuxfoundation.org/pipermail/bitcoin-dev/2015-September/011157.html},
  2015.

\bibitem{gencer2018decentralization}
A.~E. Gencer, S.~Basu, I.~Eyal, R.~van Renesse, and E.~G. Sirer.
\newblock {Decentralization in Bitcoin and Ethereum networks}.
\newblock {\em arXiv preprint arXiv:1801.03998}, 2018.

\bibitem{gervais2016security}
A.~Gervais, G.~O. Karame, K.~W{\"u}st, V.~Glykantzis, H.~Ritzdorf, and
  S.~Capkun.
\newblock On the security and performance of proof of work blockchains.
\newblock In {\em Proceedings of the 2016 ACM SIGSAC Conference on Computer and
  Communications Security}. ACM, 2016.

\bibitem{gervais2015tampering}
A.~Gervais, H.~Ritzdorf, G.~O. Karame, and S.~Capkun.
\newblock Tampering with the delivery of blocks and transactions in {B}itcoin.
\newblock In {\em Proceedings of the 22nd ACM SIGSAC Conference on Computer and
  Communications Security}. ACM, 2015.

\bibitem{karame2012double}
G.~O. Karame, E.~Androulaki, and S.~Capkun.
\newblock Double-spending fast payments in {B}itcoin.
\newblock In {\em Proceedings of the 2012 ACM conference on Computer and
  communications security}. ACM, 2012.

\bibitem{kogias2016enhancing}
E.~K. Kogias, P.~Jovanovic, N.~Gailly, I.~Khoffi, L.~Gasser, and B.~Ford.
\newblock Enhancing bitcoin security and performance with strong consistency
  via collective signing.
\newblock In {\em 25th {USENIX} Security Symposium ({USENIX} Security 16)},
  2016.

\bibitem{lewenberg2015inclusive}
Y.~Lewenberg, Y.~Sompolinsky, and A.~Zohar.
\newblock Inclusive block chain protocols.
\newblock In {\em International Conference on Financial Cryptography and Data
  Security}. Springer, 2015.

\bibitem{luu2016secure}
L.~Luu, V.~Narayanan, C.~Zheng, K.~Baweja, S.~Gilbert, and P.~Saxena.
\newblock A secure sharding protocol for open blockchains.
\newblock In {\em Proceedings of the 2016 ACM SIGSAC Conference on Computer and
  Communications Security}. ACM, 2016.

\bibitem{luu2015power}
L.~Luu, R.~Saha, I.~Parameshwaran, P.~Saxena, and A.~Hobor.
\newblock On power splitting games in distributed computation: The case of
  {B}itcoin pooled mining.
\newblock In {\em Computer Security Foundations Symposium (CSF), 2015 IEEE
  28th}. IEEE, 2015.

\bibitem{203704}
L.~Luu, Y.~Velner, J.~Teutsch, and P.~Saxena.
\newblock Smartpool: Practical decentralized pooled mining.
\newblock In {\em 26th {USENIX} Security Symposium ({USENIX} Security 17)}.
  {USENIX} Association, 2017.

\bibitem{Merkle1988}
R.~C. Merkle.
\newblock A digital signature based on a conventional encryption function.
\newblock In {\em Proceedings of Advances in Cryptology}, 1988.

\bibitem{miller2014permacoin}
A.~Miller, A.~Juels, E.~Shi, B.~Parno, and J.~Katz.
\newblock Permacoin: Repurposing {B}itcoin work for data preservation.
\newblock In {\em 2014 IEEE Symposium on Security and Privacy (SP)}. IEEE,
  2014.

\bibitem{miller2015nonoutsourceable}
A.~Miller, A.~Kosba, J.~Katz, and E.~Shi.
\newblock Nonoutsourceable scratch-off puzzles to discourage {B}itcoin mining
  coalitions.
\newblock In {\em Proceedings of the 22nd ACM SIGSAC Conference on Computer and
  Communications Security}. ACM, 2015.

\bibitem{longrunBitcoin}
M.~M{\"o}ser and R.~B{\"o}hme.
\newblock Trends, tips, tolls: A longitudinal study of bitcoin transaction
  fees.
\newblock In {\em Financial Cryptography Workshops}, 2015.

\bibitem{nakamoto2008bitcoin}
S.~Nakamoto.
\newblock Bitcoin: A peer-to-peer electronic cash system, 2008.

\bibitem{narayanan2016bitcoin}
A.~Narayanan, J.~Bonneau, E.~Felten, A.~Miller, and S.~Goldfeder.
\newblock {\em Bitcoin and cryptocurrency technologies: A comprehensive
  introduction}.
\newblock Princeton University Press, 2016.

\bibitem{TN1}
T.~Nolan.
\newblock {Distributing low POW headers}.
\newblock
  \url{https://lists.linuxfoundation.org/pipermail/bitcoin-dev/2013-July/002976.html},
  2013.

\bibitem{papagiannaki2003measurement}
K.~Papagiannaki, S.~Moon, C.~Fraleigh, P.~Thiran, and C.~Diot.
\newblock Measurement and analysis of single-hop delay on an {IP} backbone
  network.
\newblock {\em IEEE Journal on Selected Areas in Communications}, 21(6), 2003.

\bibitem{pass2017fruitchains}
R.~Pass and E.~Shi.
\newblock Fruitchains: A fair blockchain.
\newblock In {\em Proceedings of the ACM Symposium on Principles of Distributed
  Computing}. ACM, 2017.

\bibitem{poolhop}
Raulo.
\newblock {Optimal pool abuse strategy}.
\newblock \url{http://bitcoin.atspace.com/poolcheating.pdf}, 2011.

\bibitem{ritz2018impact}
F.~Ritz and A.~Zugenmaier.
\newblock The impact of uncle rewards on selfish mining in {E}thereum.
\newblock {\em arXiv preprint arXiv:1805.08832}, 2018.

\bibitem{rizun2016subchains}
P.~R. Rizun.
\newblock Subchains: A technique to scale {B}itcoin and improve the user
  experience.
\newblock {\em Ledger}, 1, 2016.

\bibitem{wBlockgoodandbad}
K.~Rosenbaum.
\newblock {Weak Blocks -- The Good And The Bad}.
\newblock
  \url{http://popeller.io/index.php/2016/01/19/weak-blocks-the-good-and-the-bad/},
  2016.

\bibitem{rosenfeld2011analysis}
M.~Rosenfeld.
\newblock Analysis of {B}itcoin pooled mining reward systems.
\newblock {\em arXiv preprint arXiv:1112.4980}, 2011.

\bibitem{wBlockSimulator}
R.~Russell.
\newblock {Weak block simulator for Bitcoin}.
\newblock \url{https://bitcointalk.org/index.php?topic=179598.0}, 2017.

\bibitem{sapirshtein2016optimal}
A.~Sapirshtein, Y.~Sompolinsky, and A.~Zohar.
\newblock Optimal selfish mining strategies in {B}itcoin.
\newblock In {\em International Conference on Financial Cryptography and Data
  Security}. Springer, 2016.

\bibitem{sompolinsky2016spectre}
Y.~Sompolinsky, Y.~Lewenberg, and A.~Zohar.
\newblock {SPECTRE}: Serialization of proof-of-work events: confirming
  transactions via recursive elections, 2016.

\bibitem{2013-Sompolinsky}
Y.~Sompolinsky and A.~Zohar.
\newblock Accelerating {B}itcoin'€™s transaction processing.
\newblock {\em Fast Money Grows on Trees, Not Chains}, 2013.

\bibitem{cryptoeprint:2018:104}
Y.~Sompolinsky and A.~Zohar.
\newblock {PHANTOM}, {GHOSTDAG}: Two scalable {BlockDAG} protocols.
\newblock Cryptology ePrint Archive, Report 2018/104, 2018.
\newblock \url{https://eprint.iacr.org/2018/104}.

\bibitem{SzCVCBT18}
P.~Szalachowski.
\newblock (short paper) towards more reliable {B}itcoin timestamps.
\newblock In {\em Proceedings of the Crypto Valley Conference on Blockchain
  Technology (CVCBT)}, 2018.

\bibitem{ethereum-white-paper}
E.~team.
\newblock {A Next-Generation Smart Contract and Decentralized Application
  Platform}.
\newblock
  \url{https://github.com/ethereum/wiki/wiki/White-Paper\#modified-ghost-implementation},
  2018.

\bibitem{nearBlockFirst}
{TierNolan (Pseudonymous)}.
\newblock {Decoupling Transactions and PoW}.
\newblock \url{https://bitcointalk.org/index.php?topic=179598.0}, 2013.

\bibitem{PT1}
P.~Todd.
\newblock {Near-block broadcasts for proof of tx propagation}.
\newblock
  \url{https://lists.linuxfoundation.org/pipermail/bitcoin-dev/2013-September/003275.html},
  2013.

\bibitem{gapgame}
I.~Tsabary and I.~Eyal.
\newblock The gap game.
\newblock In {\em Proceedings of the 2018 ACM SIGSAC Conference on Computer and
  Communications Security}. ACM, 2018.

\bibitem{wood2014ethereum}
G.~Wood.
\newblock Ethereum: A secure decentralised generalised transaction ledger.
\newblock {\em Ethereum project yellow paper}, 151, 2014.

\bibitem{zamyatinflux}
A.~Zamyatin, N.~Stifter, P.~Schindler, E.~Weippl, and W.~J. Knottenbelt.
\newblock {Flux}: Revisiting near blocks for proof-of-work blockchains.
\newblock Cryptology ePrint Archive, Report 2018/415, 2018.
\newblock \url{https://eprint.iacr.org/2018/415/20180529:172206}.

\bibitem{zeng2017nonoutsourceable}
G.~Zeng, S.~M. Yiu, J.~Zhang, H.~Kuzuno, and M.~H. Au.
\newblock A nonoutsourceable puzzle under {GHOST} rule.
\newblock In {\em 2017 15th Annual Conference on Privacy, Security and Trust
  (PST)}. IEEE, 2017.

\bibitem{zhang2019lay}
R.~Zhang and B.~Preneel.
\newblock Lay down the common metrics: Evaluating proof-of-work consensus
  protocols' security.
\newblock In {\em 2019 IEEE Symposium on Security and Privacy (SP)}. IEEE,
  2019.

\end{thebibliography}
